\documentclass[11pt]{article}

\RequirePackage{amsthm,amsmath,amsfonts,amssymb,graphicx,mathtools,enumerate,placeins}
\usepackage[round]{natbib}

\usepackage[CJKbookmarks=true,
            bookmarksnumbered=true,  
			bookmarksopen=true,
			colorlinks=true,
			citecolor=blue,
			linkcolor=blue,
			anchorcolor=red,
			urlcolor=blue]{hyperref}

\usepackage[letterpaper, left=1.2truein, right=1.2truein, top = 1.2truein, bottom = 1.2truein]{geometry}

\usepackage[blocks, affil-it]{authblk}

\theoremstyle{plain}
\newtheorem{theorem}{Theorem}

\newtheorem{lemma}[theorem]{Lemma}
\newtheorem{proposition}[theorem]{Proposition}

\theoremstyle{remark}
\newtheorem{remark}[theorem]{Remark}

\newcommand\cR{\mathcal{R}}

\newcommand\R{\mathbb{R}}
\renewcommand\P{\mathbb{P}}
\newcommand\E{\mathbb{E}}
\newcommand\eps{\varepsilon}

\DeclareMathOperator*{\argmax}{argmax}
\DeclareMathOperator*{\argmin}{argmin}

\newcommand{\iid}{\textnormal{iid}}
\newcommand{\ind}{\textnormal{ind}}
\newcommand{\simiid}{\stackrel{\iid}{\sim}}
\newcommand{\simind}{\stackrel{\ind}{\sim}}

\newcommand{\gren}{\hat{f}_m}
\newcommand{\lfdr}{\textnormal{lfdr}}

\newcommand{\maxlfdr}{\textnormal{max-lfdr}}
\newcommand{\FDR}{\textnormal{FDR}} 
\newcommand{\Regret}{\textnormal{Regret}} 
 
\newcommand{\BH}{\textnormal{BH}} 
\newcommand{\Fix}{\textnormal{Fix}}

\newcommand{\q}{q} 
\newcommand{\tol}{\ell} 

\newcommand{\edit}[1]{{#1}} 

\title{The edge of discovery: Controlling the local\\ false discovery rate at the margin}
\author[1]{Jake A. Soloff}
\author[1]{Daniel Xiang}
\author[2]{William Fithian}
\affil[1]{Department of Statistics, University of Chicago} \affil[2]{Department of Statistics, University of California, Berkeley} 

\begin{document}
\maketitle

\begin{abstract}
Despite the popularity of the false discovery rate (\FDR) as an error control metric for large-scale multiple testing, its close Bayesian counterpart the local false discovery rate (\lfdr), defined as the posterior probability that a particular null hypothesis is false, is a more directly relevant standard for justifying and interpreting individual rejections. However, the lfdr is difficult to work with in small samples, as the prior distribution is typically unknown. We propose a simple multiple testing procedure and prove that it controls the expectation of the maximum $\lfdr$ across all rejections; equivalently, it controls the probability that the rejection with the largest $p$-value is a false discovery. Our method operates without knowledge of the prior, assuming only that the $p$-value density is uniform under the null and decreasing under the alternative. We also show that our method asymptotically implements the oracle Bayes procedure for a weighted classification risk, optimally trading off between false positives and false negatives. We derive the limiting distribution of the attained maximum lfdr over the rejections, and the limiting empirical Bayes regret relative to the oracle procedure.
\end{abstract}

\section{Introduction}\label{sec-intro}

A common goal in applications of multiple hypothesis testing is to identify a relatively short list of candidate  ``discoveries'' that are sufficiently promising to undertake some costly further action. In scientific applications, for example, each discovery may be the focus of a follow-up experiment \edit{that wastes resources} if the apparent discovery was only a mirage. The {\em false discovery rate} \citep[FDR,][]{benjamini1995controlling} has become a cornerstone of modern large-scale multiple testing because it directly measures the rate of this wastage:\footnote{\edit{See \citet{benjamini2000adaptive} for a review of the history of FDR and the Benjamini-Hochberg (BH) procedure, including the work of Eklund and Seeger in the 1960s. See also \citet{seeger1968note} for early developments.}}
\begin{quote}
    [T]he proportion of errors in the
pool of candidates is of great economical significance since follow-up studies are costly, and thus avoiding multiplicity control is costly. Indeed, the $\FDR$ criterion is economically interpretable; when considering a potential threshold, the adjusted $\FDR$ gives the proportion of the investment that is about to be wasted on false leads. \citep{reiner2003identifying}
\end{quote} 
An analyst who controls $\FDR$ at level~$\q = 5\%$, then, is willing to waste resources following up on one false discovery in exchange for every nineteen real discoveries. 

Carrying this reasoning further, however, we can apply the same cost-benefit analysis to each individual rejection, not only to the list of rejections taken as a whole. In economic terminology, we should consider not only the {\em average utility} of our entire rejection set, but also the {\em marginal utility} of each rejection we make, since we always have the option to exclude any rejection that is not individually promising. For example, in Section~\ref{sec-simulations} we reproduce the simulations of~\citet{benjamini1995controlling} and find in some settings that, even while the Benjamini--Hochberg (BH) procedure controls $\FDR$ at level~$\q = 5\%$, the {\sl last discovery} (i.e. the discovery with the largest~$p$-value) is false more than~$30\%$ of the time. In such settings, unless we are willing to suffer one false discovery for every two true discoveries, we would be better served by excluding the last rejection from the BH rejection set. More generally, to decide where to set our rejection threshold, we should ask about the proportion of false leads among the incremental rejections that we would add or remove by raising or lowering it.

The likelihood that an individual discovery is a false lead is called its {\em local false discovery rate} \citep[$\lfdr$,][]{efron2001empirical}. For $i=1,\ldots,m$, let $H_i = 0$ if the $i$th hypothesis is null and $H_i = 1$ otherwise, and consider the simple {\em Bayesian two-groups model}
\begin{align}\label{eq-two-groups}
p_i \mid H_i=h \simind f_h, \qquad \textnormal{with} \qquad H_i \simiid\textnormal{Bern}(1-\pi_0), \qquad \textnormal{for } i = 1,\dots,m,
\end{align}
where $f_0$ and $f_1$ are densities (null and alternative, respectively) supported on the unit interval~$[0,1]$, and the null proportion is $\pi_0\in [0,1]$. We will assume throughout that $f_0 = 1_{[0,1]}$, the uniform density. Let $f \coloneqq \pi_0 + (1-\pi_0)f_1$ denote the common mixture density of the $p$-values in model~\eqref{eq-two-groups}, and let $F(t) \coloneqq \int_0^t f(u)\,du$ denote the corresponding cumulative distribution function (cdf). The lfdr is then defined as the posterior probability that $H_i = 0$, conditional on the observed $p$-value~$p_i$:
\begin{align}\label{eq-def-lfdr}
    \lfdr(t) \coloneqq \P\left\{H_i = 0\mid p_i = t\right\} = \frac{\pi_0}{f(t)}.
\end{align}
If we knew the problem parameters $\pi_0$ and $f_1$, then the definition~\eqref{eq-def-lfdr} would neatly solve the problem posed above: we should reject only those hypotheses whose lfdr is below the break-even threshold of our cost-benefit tradeoff. Concretely, let $\omega > 0$ define the ratio between the cost of each false discovery and the benefit of each true discovery. Then the utility of making $R$ rejections, of which $V$ are false discoveries, is proportional to $(R-V) - \omega V$, and a simple calculation shows that we should reject the $i$th hypothesis if and only if $\lfdr(p_i) \leq \alpha \coloneqq \frac{1}{1+\omega}$. 

We will usually work under the additional assumption that $f_1(t)$ is non-increasing in $t$, or equivalently that $\lfdr(t)$ is non-decreasing, so that smaller $p$-values represent stronger evidence against the null. This assumption is common in multiple testing \citep[see, e.g.,][]{genovese2004stochastic,langaas2005estimating,strimmer2008unified}, and it lets us restrict our attention to procedures that reject all $p$-values below a given threshold: if $f_1$ is non-increasing then rejecting when $\lfdr(p_i) \leq \alpha$ is equivalent to rejecting when $p_i$ is sufficiently small.

In practice, $\pi_0$ and $f_1$ are typically unknown and must be estimated from the data, and many estimators have been proposed; see e.g. \citet{efron2001empirical, pounds2003estimating, scheid2004stochastic, aubert2004determination, efron2004large, efron2008microarrays, liao2004mixture, pounds2004improving, robin2007semi, strimmer2008unified, muralidharan2010empirical, patra2016estimation, stephens2017false}. To the best of our knowledge, however, there are no known finite-sample lfdr control guarantees for multiple testing procedures based on these methods. By contrast, simple, robust, and well-known methods like the Benjamini--Hochberg (BH) procedure of \citet{benjamini1995controlling} enjoy finite-sample $\FDR$ control without requiring the analyst to model the $p$-value distribution.

In this work, we introduce a new error control metric that measures the lfdr of a multiple testing procedure's least promising rejection. We represent a generic multiple testing method as a function $\cR(p_1,\ldots,p_m)$ returning an index set $\cR \subseteq \{1,\ldots,m\}$, where hypothesis $i$ is rejected if and only if $i \in \cR$. We say the procedure's {\em max-lfdr} is
\begin{equation}\label{eq:def-max-lfdr}
    \maxlfdr(\cR) \coloneqq
    \E\left[\max_{i\in \cR}\;\lfdr(p_i)\right],
\end{equation}
defining the maximum as zero if no rejections are made. 

We can consider the $\maxlfdr$ as a frequentist error control criterion in the two-groups model \eqref{eq-two-groups}, which may not be a fully Bayesian model if we treat $\pi_0$ and $f_1$ as unknown. If $f_1$ is non-increasing, then the $\maxlfdr$ of $\cR$ coincides with the probability that the last rejection is a false discovery. This latter definition extends beyond the two-groups model, to the setting where the Bernoulli variables $H_1,\ldots,H_m$ are fixed rather than random.

In addition to the $\maxlfdr$ criterion, we also introduce a simple multiple testing procedure, the {\em support line} (SL) procedure, and show that it provably controls the max-lfdr under mild assumptions. Define the $p$-value order statistics $p_{(1)} \leq \cdots \leq p_{(m)}$, and let $p_{(0)} = 0$ by convention. Then our procedure rejects $p$-values up to the last minimizer
\begin{align}\label{eq-def-procedure}
R_{\tol} \coloneqq \argmin_{k=0,\ldots,m}\; p_{(k)} - \frac{\tol k}{m}.
\end{align}
That is, we reject $\cR_{\tol} \coloneqq \{i:\; p_i \leq \tau_{\tol}\}$, for the threshold $\tau_{\tol} = p_{(R_{\tol})}$. Under the two-groups model~\eqref{eq-two-groups}, with non-increasing $f_1$, we show in Theorem~\ref{thm-main} that, for~$\tol\le 1$,
\[
\maxlfdr(\cR_\tol) = \pi_0 \tol.
\]
Our method can be implemented without knowing $\pi_0$ or $f_1$, apart from the shape constraint, and bears a close relationship to the BH procedure, which replaces $R_{\tol}$ in \eqref{eq-def-procedure} with
\[
    R^{\BH}_{\q} \coloneqq \max\left\{k\in \{0,\ldots,m\} : p_{(k)} \leq \frac{\q k}{m} \right\},
\]
rejecting $\cR_{\q}^{\BH} \coloneqq \{i:\; p_i \leq \tau_{\q}^{\BH}\}$, for $\tau_{\q}^{\BH} \coloneqq \q R_{\q}^{\BH}/m \geq p_{(R_{\q}^{\BH})}$. \edit{The BH method makes at least as many rejections as the SL method at the same level~$\q = \tol$, i.e. $R_{\tol} \leq R_{\tol}^{\BH}$; in this case,} both methods make at least one rejection if and only if $p_{(k)} \leq \frac{\tol k}{m}$ for some $k \geq 1$. However, as we will argue, the SL method should generally be run with a strictly larger \edit{level $\tol > \q$} than we would use for BH. Figure~\ref{fig-method} illustrates the relationship between the two methods by reproducing the familiar plot of the BH procedure as an operation on the order statistics $p_{(1)}, \ldots, p_{(m)}$.

\begin{figure}[t!]
	\centering
	\centerline{
    \includegraphics{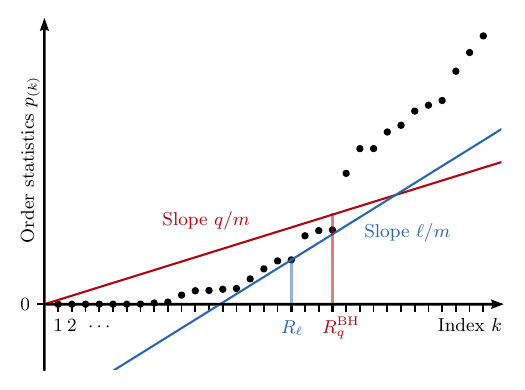} 
	}
    \caption{The order statistics~$p_{(k)}$ of the $p$-values as a function of the index~$k$, shown in black. The BH procedure, in red, finds the last index~$R_\q^\BH$ such that~$p_{(R_\q^\BH)}$ falls below the ray of slope~$\q/m$; by contrast, our procedure finds the last boundary point~$(R_\tol, p_{(R_\tol)})$ of the supporting line of slope~$\tol/m$.}\label{fig-method}
\end{figure}

\subsection{Multiple testing and the weighted classification loss}

To formalize our analysis above, define the per-instance {\em weighted classification loss}:
\begin{align}\label{eq-loss}
 L_\omega(H, \cR) \coloneqq \frac{(1+\omega)V - R}{m},
\end{align}
\edit{where $R = |\cR|$ denotes the number of rejections and $V = \sum_{i\in \cR}(1-H_i)$ denotes the number of false discoveries.} This loss can be derived, up to additive and multiplicative constants, by viewing each of the $m$ hypotheses as a binary classification problem, where we incur a cost $c_1$ for each type I error or false discovery ($i\in \cR$, but $H_i = 0$), and cost $c_2$ from each type II error or false non-discovery ($i \notin \cR$, but $H_i = 1$). If the total number of non-nulls is $m_1 = \sum_i H_i$, then there are $m_1 - (R - V)$ false non-discoveries, so the total loss over all $m$ instances is
\[
c_1 V + c_2 (m_1 - (R - V)) = c_2 m  \cdot L_\omega(H, \cR) +  c_2 m_1,
\]
where $\omega = c_1/c_2$ is the ratio between the two misclassification costs. $L_\omega$ as defined in \eqref{eq-loss} is normalized so that rejecting nothing incurs zero loss, and each true discovery has value $1/m$.

Under the two-groups model \eqref{eq-two-groups}, \citet[][Theorem 2]{sun2007oracle} show that the corresponding Bayes risk~$\E L_\omega(H, \cR)$ is minimized by the oracle procedure
\begin{equation}\label{eq-def-oracle}
\cR^* \coloneqq \left\{i:\;\lfdr(p_i) \leq \alpha\right\}, \quad \text{where} \quad \alpha = \frac{1}{1+\omega}.
\end{equation}
The ratio $\omega$ specifies the ``break-even exchange rate'' at which we are willing to trade true discoveries for false leads; e.g., if $\omega = 19$ then we are willing to suffer a single false discovery for exactly $19$ true discoveries, and we should reject a hypothesis only if its $\lfdr$ falls below the break-even tolerance $\alpha = 0.05$. If $f_1$ is non-increasing, then the oracle procedure reduces to thresholding $p$-values at a fixed threshold
\begin{align}\label{eq-oracle-threshold}
    \cR^* = \{i:\; p_i \leq \tau^*\}, \quad \text{for}\quad \tau^* \coloneqq \max\{t\in [0,1]:\; \lfdr(t) \leq \alpha\},
\end{align}
with $\tau^* = 0$ if no such threshold exists.

Our method can be directly interpreted as minimizing an empirical proxy of the weighted classification loss. For a candidate threshold $t \in [0,1]$, the expected number of null $p$-values below the threshold is $m \pi_0 t$. If $\pi_0$ is known, we can estimate $V \approx m \pi_0 t$ to obtain a running estimator of the loss from thresholding $p$-values at $t$:
\begin{equation}\label{eq-running-est}
\hat{L}_{\omega}(t; \pi_0) = \frac{(1 + \omega) m \pi_0 t - m F_m(t)}{m} = (1 + \omega) \left(\pi_0 t - \alpha F_m(t)\right),
\end{equation}
where $F_m(t)$ represents the empirical cumulative distribution function (ecdf) of the $p$-values:
\[
F_m(t) \coloneqq \frac{1}{m} \sum_{i=1}^m 1\{p_i \leq t\}.
\]
Because $\hat{L}_{\omega}(t; \pi_0)$ is continuously increasing except at the order statistics, it is minimized at one of the order statistics, or at $p_{(0)}=0$:
\[
\argmin_{k = 0, 1, \ldots, m} \hat{L}_{\omega}(p_{(k)}; \pi_0) 
\;=\; \argmin_{k = 0, 1, \ldots, m} \pi_0 p_{(k)} - \frac{\alpha k}{m}.
\]
Comparing the last expression to the definition of our procedure in~\eqref{eq-def-procedure}, we see that $\hat{L}_{\omega}(t; \pi_0)$ is minimized at $t=\tau_{\tol}$ for $\tol=\alpha/\pi_0$. By Theorem~\ref{thm-main}, we then have $\maxlfdr(\cR_{\tol}) \le \alpha$, with equality as long as~$\alpha\le \pi_0$.

By contrast, $\tau^{\BH}_{\q}$ for $\q=\alpha/\pi_0$ is the largest value of $t$ that gives $\hat{L}_{\omega}(t; \pi_0) = 0$, the same loss we would achieve by rejecting nothing at all. In other words, the BH procedure at level $\alpha/\pi_0$ only aims to break even; to do better, we should run BH at a strictly smaller level $\q < \alpha/\pi_0$. Thus, we view $\q$ as a tuning parameter whose correspondence to the cost ratio $\omega$ is generally unknown.

To select $\tol$ for our SL procedure when $\pi_0$ is unknown, we can either conservatively bound $\pi_0 \leq 1$ and run the procedure at $\tol=\alpha$, or estimate $\pi_0$ and use $\tol=\alpha/\hat\pi_0$ \edit{(see Section~\ref{sec-estimating-pi0})}. To avoid confusion, we will always use the notation $\tol$ to represent \edit{the SL procedure's} tuning parameter \edit{and $\q$ to represent the BH procedure's tuning parameter, reserving} $\alpha = \frac{1}{1+\omega}$ to represent the true target $\lfdr$, defined in terms of the cost ratio $\omega$.

Our procedure can alternatively be derived as a plug-in maximum likelihood estimator (MLE) of the oracle procedure $\cR^*$, where we estimate $f(t)$ using Grenander's nonparametric MLE for a non-increasing density \citep{grenander1956theory}:
\begin{align}\label{eq-gren}
\gren \coloneqq \argmax_{\substack{g : [0,1]\to \R_+ \\ \text{ non-increasing density}}} ~\frac{1}{m}\sum_{i=1}^m \log g(p_i).
\end{align}
As we will see in Section~\ref{sec-grenander}, $\tau_\tol$ is also the largest value $t \in [0,1]$ for which $\gren(t) \geq \tol^{-1}$. Thus, if we run our procedure at $\tol=\alpha/\pi_0$, we have
\[
\cR_{\alpha/\pi_0}  = \left\{i:\; \gren(p_i) \geq (\alpha/\pi_0)^{-1}\right\} = \left\{i:\; \frac{\pi_0}{\gren(p_i)} \leq \alpha\right\}.
\]
As above, if $\pi_0$ is unknown, we can either estimate it or conservatively bound $\pi_0 \leq 1$.

The relationship between our method and the Grenander estimator is convenient for asymptotic analysis because the latter is very well studied; see the book by \citet{groeneboom2014nonparametric} for a thorough treatment. The Grenander estimator has previously been considered for estimating the~$\lfdr$ \citep{strimmer2008unified} as well as for estimating the null proportion~$\pi_0$ \citep{langaas2005estimating}. \edit{The density estimator}~$\gren$ may be efficiently computed via the pool adjacent violators algorithm~\citep[see, e.g.,][Chapter 1]{RWD88}\edit{, but the definition of $R_\tol$ in \eqref{eq-def-procedure} provides a way to characterize and run our procedure without explicitly computing~$\gren$}.

\subsection{The max-lfdr and the FDR}

The $\maxlfdr$ in \eqref{eq:def-max-lfdr} and the $\FDR$ are two different error criteria that both appeal to the logic of trading off true and false discoveries. The key difference is that the $\FDR$, defined as
\[
\FDR(\cR) \coloneqq \E\left[\frac{V}{R} \cdot 1\{R > 0\}\right],
\]
measures the likelihood that a {\em randomly selected} rejection is null, whereas the $\maxlfdr$~\eqref{eq:def-max-lfdr} instead measures the likelihood that the {\em least promising} rejection is null. In both cases the event in question is deemed not to have occurred if $R=0$, so that under the global null (all $H_i=0$, almost surely), both criteria reduce to the probability of making a single rejection.

Throughout this section, we will restrict our attention to procedures that reject the $R$ hypotheses with the smallest $p$-values. That is, we assume a procedure $\cR$ rejects $H_{(1)}, \ldots, H_{(R)}$, where $H_{(k)}$ represents the hypothesis corresponding to $p_{(k)}$. If $f_1$ is non-increasing, then the procedure's {\em last rejection} $H_{(R)}$ is the least promising, and the $\maxlfdr$ can be equivalently characterized as the probability that the last rejection is a false discovery:
\begin{equation}\label{eq:last-rejection}
\maxlfdr(\cR) = \E\left[\lfdr\left(p_{(R)}\right) \cdot 1\{R > 0\} \right] = \P\left\{H_{(R)} = 0, R > 0\right\}.
\end{equation}
If $\maxlfdr(\cR) > \alpha = \frac{1}{1+\omega}$, then we can improve $\cR$ by excluding its last discovery.\footnote{Without the shape constraint on $f_1$, $\maxlfdr > \alpha$ still implies that the analyst could improve the procedure by removing the least promising rejection, which may not be the same as the last rejection. However, this improvement is only feasible if the analyst can recognize which rejection is least promising.} Let $\cR^{(-1)}$ denote the procedure that makes one fewer rejection than $\cR$, meaning it rejects $H_{(1)}, \ldots, H_{(R-1)}$ if $R > 0$, and makes no rejections if $R=0$. Then we have
\begin{align*}
\E[L_{\omega}(H, \cR) - L_{\omega}(H, \cR^{(-1)})] 
&\;=\; \frac{1}{m}\E\left[(1 + \omega) 1\{H_{(R)} = 0, R>0\} -  1\{R > 0\}\right]\\[5pt]
&\;=\; \frac{1+\omega}{m} \left( \maxlfdr(\cR) - \alpha \P\{R>0\} \right),
\end{align*}
which is positive if $\maxlfdr(\cR) > \alpha$. The converse, that dropping the last rejection does not improve the risk if $\maxlfdr(\cR) \leq \alpha$, is almost true if $\P\{R>0\} \approx 1$, but is not true in general: under the global null, for example, any procedure is improved by making fewer rejections.

This thought experiment --- what if we dropped the last rejection? --- is at the heart of our motivation for proposing the $\maxlfdr$ as an error criterion. Even when a rejection set's average quality is high, the rejections near the threshold may be recognizably bad bets. In that case, we are better off pruning our rejection set until all of the rejections that remain are individually worth following up on. 
Because $\maxlfdr(\cR) \geq \FDR(\cR)$, controlling the $\maxlfdr$ is more conservative than controlling $\FDR$  at the same level \edit{$\q = \tol$}, in many cases considerably so. From this, it is tempting to conclude that $\maxlfdr$ control is an inherently more conservative goal than $\FDR$ control, but this conclusion would be mistaken. An analyst whose break-even exchange rate is $\omega = 9$ and break-even tolerance is $\alpha = 0.1$, for example, would never choose a method with a $10\%$ $\FDR$; the resulting rejection set would be no better on average than rejecting nothing at all, so there would be no point in collecting the data in the first place. Thus, an analyst who is satisfied with a $10\%$ $\FDR$ must have a larger break-even tolerance, say $\alpha = 0.2$ or $0.3$.

By the same token, it would be unfair to evaluate the risk under $L_\omega$ of the $\BH$ procedure at level $\q = \alpha = \frac{1}{1+\omega}$, since an analyst whose break-even tolerance is $\alpha$ would  want to control $\FDR$ at a strictly smaller level $\q$, like $\alpha/2$ or $\alpha/10$. However, as we show in Section~\ref{sec-popn-regret}, the performance of $\BH(\q)$ with such {\em a priori} choices of $\q$ can depend sensitively on the unknown alternative density $f_1$.

\subsection{Outline and contributions}

In Section~\ref{sec-control}, we state and prove our main result, that $\maxlfdr(\cR_\tol) = \pi_0\tol$ under the Bayesian two-groups model with non-increasing $f_1$. Even without monotonicity of $f_1$, we have $\P\{H_{(R_{\tol})} = 0, R_{\tol}>0\} = \pi_0\tol$, but monotonicity ensures that the lfdr is not out of control for rejections in the interior of the rejection region. We also prove $\maxlfdr$ control for an adaptive method that estimates $\pi_0$ from the data in the same way as the procedure of \citet{storey2004strong}.

In Section~\ref{sec-regret}, we investigate our method's asymptotic performance relative to the oracle procedure $\cR^*$. Extending asymptotic results for the Grenander estimator, we show that our method's attained lfdr threshold, $\lfdr(\tau_\tol)$, concentrates at a rate $m^{-1/3}$ around its expectation $\pi_0 \tol$, giving an explicit formula for its asymptotic distribution. We also show that our method's asymptotic regret relative to the oracle shrinks at the rate $m^{-2/3}$. Section~\ref{sec-simulations} illustrates our results with selected simulations, and Section~\ref{sec-discussion} concludes.

\section{Finite-sample max-lfdr control}\label{sec-control}

\subsection{Main result}

Our main result is that our procedure~$\cR_\tol$ controls the $\maxlfdr$ at exactly $\pi_0 \tol$. 

\begin{theorem}\label{thm-main} Suppose $p_1,\ldots,p_m$ follow the Bayesian two-groups model~\eqref{eq-two-groups}, with uniform null density $f_0 = 1_{[0,1]}$. For the procedure defined in \eqref{eq-def-procedure} with $\tol \leq 1$, we have
\begin{align}\label{eq-lfdr-control}
\E\left[\lfdr\left(p_{(R_{\tol})}\right) \cdot 1\{R_{\tol} > 0\} \right] = \P\left\{H_{(R_{\tol})} = 0, R_{\tol} > 0\right\} = \pi_0\tol.
\end{align}
Furthermore, if the alternative density $f_1$ is non-increasing, then we have 
\[
\maxlfdr(\cR_{\tol}) = \pi_0\tol.
\]
\end{theorem}

The familiar optional-stopping arguments from the $\FDR$ control literature, introduced by \citet{storey2004strong}, do not seem to apply to our procedure, since the minimizer~$R_{\tol}$ of the sequence~$p_{(k)} - \tol k/m$ for $k=0,\ldots,m$ is not a stopping time \edit{in the usual filtration}. We instead prove Theorem~\ref{thm-main} via a conditioning argument, whose crux is showing that each null $p$-value has exactly an $\tol/m$ chance of being the last rejection $p_{(R_{\tol})}$:

\begin{lemma}\label{lem-last-rejection}
    Fix $p_1,\ldots,p_{m-1} \in [0,1]$ and let $p_m \sim \textnormal{Unif}(0,1)$. Then we have
    \[
    \P\{p_{(R_\tol)} = p_m\} \leq \tol/m,
    \]
    with equality if $\tol \leq 1$.
\end{lemma}
Given Lemma~\ref{lem-last-rejection}, the proof of Theorem~\ref{thm-main} is straightforward:

\begin{proof}[Proof of Theorem~\ref{thm-main}] The first equality in~\eqref{eq-lfdr-control} follows from conditioning on~$\{p_i\}_{i=1}^m$, since
\begin{align*}
\P\left\{H_{(R_{\tol})} = 0, R_{\tol} > 0\mid p_1,\ldots,p_m\right\} 
&\;=\; \sum_{i=1}^m \P\left\{H_i = 0\mid p_1,\ldots,p_m\right\}1\{p_i = p_{(R_\tol)}\}\\[5pt]
&\;=\; \lfdr(p_{(R_{\tol})})\cdot 1\{R_{\tol} > 0\}.
\end{align*}
Next, because the $(H_i, p_i)$ pairs are independent and identically distributed, we can decompose the probability in \eqref{eq-lfdr-control} as
\begin{align*}
\P\left\{H_{(R_{\tol})} = 0, R_{\tol} > 0\right\} 
&\;=\; \sum_{i=1}^m \P\left\{H_i = 0, p_{(R_{\tol})} = p_i\right\}\\[5pt]
&\;=\; m \P\left\{H_m = 0, p_{(R_{\tol})} = p_m\right\}\\[5pt]
&\;=\; \pi_0 m \P\left\{p_{(R_{\tol})} = p_m \mid H_m = 0\right\}\\[5pt]
&\;=\; \pi_0 \tol,
\end{align*}
where the last step comes from conditioning on $p_1,\ldots,p_{m-1}$ and applying Lemma~\ref{lem-last-rejection}. If $f_1(t)$ is non-increasing, then $\lfdr(t)$ is non-decreasing, so that $\max_{i \in \cR_{\tol}} \lfdr(p_i) = \lfdr(p_{(R_\tol)})$ almost surely, completing the argument.
\end{proof}

We now turn to proving Lemma~\ref{lem-last-rejection}. Because $p_m$ is uniform, the probability statement is equivalent to a showing that, for any fixed $p_1,\ldots,p_{m-1} \in [0,1]$, the set of ``winning values'' $p_m \in [0,1]$, for which $\tau_{\tol}(p_1,\ldots,p_m) = p_m$, has Lebesgue measure $\tol/m$. 

\begin{proof}[Proof of Lemma~\ref{lem-last-rejection}] As we hold $p^{-m} = (p_1,\ldots,p_{m-1})$ fixed and vary~$p_m\in [0,\infty)$, define the attained minimum of the loss estimator as
\[
\varphi(p_m) \coloneqq \min_{k=0,\ldots,m}p_{(k)}-\tol\frac{k}{m}.
\]
\edit{For $k=1,\ldots,m$, let $\varphi_k(p_m) = p_{(k)}-\tol\frac{k}{m}$, and $\varphi_0 \equiv 0$, so that $\varphi(p_m) = \min_{k=0,\ldots,m} \varphi_k(p_m)$. Each $\varphi_k$ is continuous, non-decreasing, and piecewise linear with at most three pieces. Specifically, for $k=1,\ldots,m-1$, $\varphi_k$ has slope 1 on the open interval $p_m \in (p_{(k-1)}^{-m}, p_{(k)}^{-m})$ where $p_m = p_{(k)}$, and slope 0 elsewhere (if $p_{(k-1)}^{-m} = p_{(k)}^{-m}$ then $\varphi_k$ is constant). For $k=m$, $\varphi_k$ has slope 1 on the open interval $(p_{(m-1)}^{-m}, \infty)$ and slope 0 elsewhere. As a result, $\varphi(p_m)$ is also continuous, non-decreasing, and piecewise linear, with finitely many knots. Between its knots, the function's slope is 1 on the intervals where the minimizing function is increasing, and slope 0 everywhere else. Furthermore, because the minimizing function is $\varphi_{R_{\tol}}$, the region where $\varphi(p_m)$ is differentiable and $p_m = p_{(R_{\tol})}$ is exactly the region where $\varphi'(p_m) = 1$.}

By the fundamental theorem of calculus,
\begin{align*}
    \P\{p_{(R_\tol)} = p_m\} 
    &=\int_0^1 1\{p_{(R_\tol)} = p_m\} \textnormal{d}p_m \\
    &=\int_0^1 \varphi'(p_m) \textnormal{d}p_m \\
    &= \varphi(1)-\varphi(0).
\end{align*}
It remains only to evaluate $\varphi(1) - \varphi(0)$. \edit{To complete the argument informally, note that moving $p_m$ from $0$ to $1$ shifts all order statistics by one index. If $\tol < 1$ this results in an identical rejection threshold but with one fewer rejection, increasing $\varphi$ by exactly $\tol/m$. If ${\tol \ge 1}$, however, this argument is not quite correct because the rejection threshold could possibly increase to $p_m=1$.}

More formally, let $p_{(k)}(u)$ and $R_{\tol}(u)$ represent the order statistics and number of rejections when we set $p_m = u \in [0,\infty)$. For any $u > \max\{\tol,1\}$, we have $p_{(k-1)}(u) = p_{(k)}(0)$ for all $k < m$, and $p_{(m)}(u) = u > \tol\frac{m}{m}$. As a direct result, we have $R_{\tol}(u) = R_{\tol}(0) - 1$ and $p_{(R_{\tol}(u))}(u) = p_{(R_{\tol}(0))}(0)$, so we have $\varphi(u) - \varphi(0) = \tol/m$. By the continuity and monotonicity of $\varphi$, we also have
\[
\varphi(1) - \varphi(0) \leq \varphi\left(\max\{\tol,1\}\right) - \varphi(0) = \tol/m.
\]
In particular, we have equality if $\tol \leq 1$, completing the proof.
\end{proof}

\begin{remark}[Extending Theorem~\ref{thm-main} to more general null densities~$f_0$]
 Because the set of ``winning values'' in Lemma~\ref{lem-last-rejection} is a subset of $[0,\tol]$ with Lebesgue measure $\tol/m$, we can trivially extend the result to conclude $\P\{p_{(R_{\tol})} = p_{m}\} \leq \tol/m$, if $p_m$ is drawn from any density $f_0$ with $f_0(t) \leq 1$ for all $t\in [0,\tol]$. Likewise, we can extend Theorem~\ref{thm-main} to show that $\maxlfdr(\cR_{\tol}) \leq \pi_0\tol$ with a more general null density $f_0$, as long as $\lfdr(t)$ is non-decreasing and $f_0(t) \leq 1$ for all $t\in [0,\tol]$.
\end{remark}

\subsection{Estimating \texorpdfstring{$\pi_0$}{the null proportion}}\label{sec-estimating-pi0}

Theorem~\ref{thm-main} parallels the exact $\FDR$ guarantee~$\FDR(\cR^{\BH}_\q) = \pi_0\q$ for the BH procedure. If we bound $\pi_0 \leq 1$, we can run our method at level $\tol=\alpha$ and ensure that we conservatively control $\maxlfdr$ at $\pi_0\alpha$, but our method will be overly conservative. 
\edit{In this section, we consider modifications of our procedure analogous to the modifications of the BH procedure proposed by \citet{storey2004strong}. First, we use their estimator of the null proportion $\pi_0$, defined as
\begin{align}\label{eq-def-storey}
    \hat{\pi}_0^\lambda 
\coloneqq \frac{1+\#\{i : p_i > \lambda\}}{(1-\lambda)m},
\end{align}
modifying an estimator originally proposed by \citet{schweder1982plots} and later by \citet{storey2002direct}. Next, we constrain the procedure to minimize over order statistics~$p_{(k)}$ that are less than~$\lambda$---see~\eqref{eq-def-procedure-modification} for the modified SL procedure.}

Our next result shows that plugging in $\hat\pi_0^\lambda$ and running a modification of our procedure at level $\hat\tol=\alpha/\hat\pi_0^\lambda$ controls $\maxlfdr$ at level $\alpha$ in finite samples:

\begin{theorem}\label{thm-main-modification} Suppose $p_1,\ldots,p_m$ follow the Bayesian two-groups model~\eqref{eq-two-groups}, with $f_0 = 1_{[0,1]}$ and~$f_1$ non-increasing. Fix~$\lambda\in(0, 1)$, and define a modified version of our SL procedure that only examines order statistics below $\lambda$:
\begin{equation}\label{eq-def-procedure-modification}
    R_{\tol}^{\lambda} \coloneqq \argmin_{k\geq 0: \;p_{(k)}\le \lambda}\,\hat\pi_0^\lambda p_{(k)} - \frac{\tol k}{m},
\end{equation}
and $\cR_{\tol}^\lambda = \{i:\; p_i \leq p_{(R_{\tol}^{\lambda})}\}$. Then we have 
\[
\maxlfdr\left(\cR_{\tol}^{\lambda}\right) 
\;\leq\; \tol.
\]
\end{theorem}
The proof of Theorem~\ref{thm-main-modification} is deferred to the Appendix. The method $\cR_\alpha^\lambda$ coincides with $\cR_{\alpha/\hat\pi_0^\lambda}$, our original procedure applied at the corrected level $\hat{\tol}=\alpha/\hat\pi_0^\lambda$, whenever $\tau_{\hat{\tol}} \leq \lambda$. Since we usually have $\tau_{\hat{\tol}} \ll 0.5 \leq \lambda$, the two methods are identical for all practical purposes.

In the next section, we will investigate the asymptotic regret of methods that estimate $\pi_0$. In particular, we will show that this estimation error is asymptotically negligible if it shrinks at a faster rate than $m^{-1/3}$. We can indeed achieve this with $\hat\pi_0^\lambda$ if $f_1$ has two continuous derivatives in a neighborhood of $1$, with $f_1'(1) = f_1(1) = 0$. By Taylor's theorem, we have
\[
1-F(\lambda) \;=\; (1-\lambda)\pi_0 + \frac{(1-\pi_0)f_1''(\xi)}{6} (1-\lambda)^3,
\]
for some $\xi \in [\lambda,1]$. Assuming $\pi_0\in (0,1)$ and taking $\lambda = 1-m^{-1/5}$, we then have
\begin{equation}\label{eq-pi0-asy}
m^{2/5}\left(\hat{\pi}_0^\lambda - \pi_0\right)
\;\sim\; m^{2/5}\left(\frac{1+\text{Binom}\left(m, 1-F(\lambda)\right)}{(1-\lambda)m} - \pi_0\right)\\[5pt]
\;\stackrel{d}{\to}\; \mathcal{N}\left(\frac{(1-\pi_0)f_1''(1)}{6}, \pi_0\right),
\end{equation}
with subgaussian errors for finite $m$, so the results in Section~\ref{sec-ao} generally apply. \edit{Here we use the asymptotic approximation for $\epsilon \to 0$, $m\epsilon \to \infty$:
\begin{equation}
(m\epsilon)^{-1/2}\left(\text{Binom}(m,\epsilon) - m\epsilon\right) \;\stackrel{d}{\to}\; \mathcal{N}\left(0, 1\right),
\end{equation}
for $\epsilon = 1-F(\lambda)$, and we apply Slutsky's theorem.}

\edit{Other estimators of the null proportion~$\pi_0$ may be more robust to violations of the independence assumption---see, e.g., \citet{benjamini2006adaptive} for a two-stage BH procedure.} See also \citet{genovese2004stochastic} and \citet{patra2016estimation} for a discussion of estimators for~$\pi_0$.

\section{Asymptotic regret analysis}\label{sec-regret}

In this section, we study our procedure's empirical Bayes regret under the weighted classification risk $\E\, \left[L_\omega(H, \cR)\right]$, where the expectation is taken over $H_1,\ldots,H_m$ and $p_1,\ldots,p_m$ according to~\eqref{eq-two-groups}, and $L_\omega$ is defined as in~\eqref{eq-loss}. Throughout this section we will be considering a sequence of problems with $m\to \infty$.

A fundamental result of~\citet{sun2007oracle} is that the oracle~\eqref{eq-def-oracle} minimizes the weighted classification risk over all procedures, thus representing a benchmark against which we can compare methods that are feasible without {\em a priori} \edit{knowledge of the} $\lfdr$. In the empirical Bayes literature~\citep[see, e.g.,][]{efron2019bayes}, the price of our ignorance of the model parameters is measured by the {\sl regret}, or excess risk, given by the optimality gap
\begin{align}\label{eq-def-regret}
\Regret_m(\cR) \coloneqq \E\left[L_\omega(H, \cR) - L_\omega\left(H, \cR^*\right)\right].
\end{align}

\subsection{Population regret}\label{sec-popn-regret}

Before tackling the more delicate problem of calculating the regret for procedures with data-dependent $p$-value rejection thresholds, we first investigate the regret of fixed-threshold methods. For $t\in [0,1]$, let $\cR_{t}^{\Fix} \coloneqq \{i:\; p_i \leq t\}$, and note that the oracle method is $\cR^* = \cR_{\tau^*}^\Fix$ where~$\tau^*$ is the oracle threshold~\eqref{eq-oracle-threshold}. We introduce the function $\rho(t)$ to represent the regret of this method, which is free of $m$:
\begin{equation}\label{eq-regret-fixed-t}
    \rho(t) \;\coloneqq\; \Regret_m(\cR_t^{\Fix})
    \;=\; F(\tau^*) - F(t) - \frac{\pi_0}{\alpha}(\tau^* - t).
\end{equation}
If $\lfdr(\tau^*) = \alpha$, then we also have $f(\tau^*) = \pi_0/\alpha$, and $\rho(t)$ is simply the error of the first-order Taylor expansion of $F$ around $\tau^*$, also known as the Bregman divergence associated with~$-F$. If $f$ is continuously differentiable between $t$ and $\tau^*$, then
\begin{align}\label{eq-regret-fixed-t-taylor}
\rho(t)
\;=\; \frac{-f'(\xi_t)}{2} \left(t-\tau^*\right)^2, \quad\text{ for some }  \xi_t \text{ between } t \text{ and } \tau^*.
\end{align}
Since~$F$ is concave, $\rho(t) \geq 0$. Finally, we can also rewrite \eqref{eq-regret-fixed-t} as an integral
\begin{align}
    \rho(t)
    \;=\; \int_{t}^{\tau^*} \left(1-\alpha^{-1}\lfdr(\tau)\right)\text{d}F(\tau).
\end{align}
This form for the regret underscores the relationship between the~$\lfdr$ and the regret, and will prove useful for analyzing the regret with data-dependent thresholds.

We can evaluate $\rho$ to investigate the regret of population versions of our procedure and the BH procedure, i.e. versions of the procedures with rejection thresholds chosen using the true cdf~$F$ in place of the empirical cdf~$F_m$. The population BH threshold at an arbitrary level~$\q\in (0,1)$ is found by intersecting~$F$ with the ray of slope~$\q^{-1}$, i.e.
\[
t^{\BH\textnormal{-POP}}_\q \coloneqq \sup\left\{ t\in[0,1] :\;  F(t) - t/\q = 0\right\}.
\] 
By comparison, the population version of our procedure~$\tau_\tol$ is
\[
t_{\tol} \coloneqq \sup\left\{ t\in[0,1]:\; \edit{t\in \argmax_{s\in [0,1]}\left\{F(s) - s/\tol\right\}}\right\},
\]
which coincides with the oracle threshold $\tau^*$ when $\tol=\alpha/\pi_0$. Note that $t_\tol$ is equivalent to the population BH threshold~$t^{\BH\textnormal{-POP}}_{\q^*(\tol)}$ at the lower level
\begin{align}\label{eq-bh-equivalence}
\q^*(\tol) \coloneqq \frac{t_{\tol}}{F(t_{\tol})}.    
\end{align}
Thus, there is always {\em some} value $\q^*(\tol)$ for which the BH procedure approximately reproduces the oracle, namely $\q^*(\alpha/\pi_0) = t_{\alpha/\pi_0}/F(t_{\alpha/\pi_0})$, but generally we cannot use it unless we know $f_1$ and $\pi_0$.

To illustrate the population regret in a concrete example, we consider a parametric alternative distribution
\[
f_1(t; \theta) \coloneqq \theta t^{\theta-1} \qquad \textnormal{for some } \theta\in(0,1),
\]
which is a $\textnormal{Beta}(\theta,1)$ density. This form is called a {\sl Lehmann alternative} in the multiple testing literature \citep[see, e.g.,][]{pounds2003estimating}. In this case, the population procedures at level~$\tol\in (0,1)$ use rejection thresholds
\begin{align*}
    t_{\tol}
    \;=\; \left(\frac{\tol^{-1} - \pi_0}{(1-\pi_0)\theta}\right)^{-\frac{1}{1-\theta}}, \quad\textnormal{ and} \quad
    t^{\BH\textnormal{-POP}}_\q
    \;=\; \left(\frac{\q^{-1}-\pi_0}{1-\pi_0}\right)^{-\frac{1}{1-\theta}}.
\end{align*}
Furthermore, the threshold equivalence~\eqref{eq-bh-equivalence} gives 
\[
\q^*(\tol)
\;=\; \frac{\theta \tol}{1 - (1-\theta)\pi_0 \tol} \;\approx\; \theta \tol,
\]
where the approximation holds for small values of $\tol$. Thus, the correspondence between $\tol$ and $\q^*(\tol)$ depends on the parameter $\theta$, which controls the signal strength under the alternative. For small values of $\theta$, the signal is very strong, and the ``correct'' choice of $\q$ is much smaller than the desired $\maxlfdr$ level $\alpha$, but for weaker signals (larger $\theta$), we should choose $\q$ closer to $\alpha$. \edit{Intuitively, when the signal is very strong, we expect the average rejection of the optimal procedure to be much more promising than rejections near the optimal threshold~$\tau^*$; correspondingly, the optimal procedure's FDR is much lower than~$\alpha$.} Without knowing the signal strength in advance, it is difficult to know at what values of $\q$ the BH method will perform well.

\begin{figure}[t!]
    \centering
	\centerline{
	\includegraphics{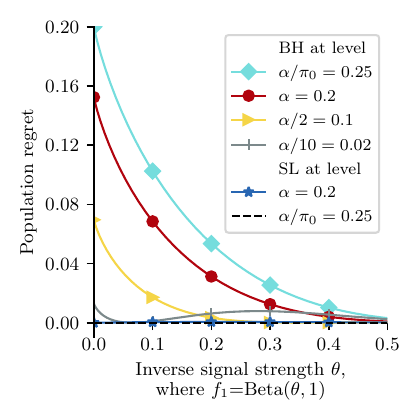}
	\includegraphics{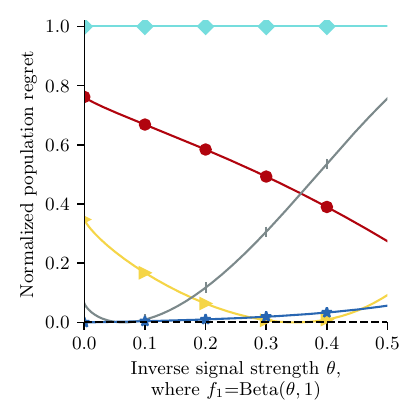}
	}
	\caption{Left: The fixed-threshold regret~$\rho(t)$~\eqref{eq-regret-fixed-t} with Beta alternatives~$f_1(s)= \theta s^{\theta-1}$ as a function of~$\theta\in [0, .5]$. Right:~a normalized version~$\rho(t)/\rho(0)$, such that BH at level~$\q = \alpha/\pi_0$ has unit normalized regret, identical to the regret of the procedure that rejects nothing. The null proportion is~$\pi_0 = 0.8$ and the cost ratio is~$\omega = 4$.}\label{fig-pop-regret}
\end{figure}

In Figure~\ref{fig-pop-regret} we plot the population regret for various choices of the level of the procedure, setting $\pi_0=0.8$ and $\omega = 4$ and varying the parameter $\theta$. The population version of our procedure at level~$\tol=\frac{\alpha}{\pi_0}$ with $\alpha= \frac{1}{1+\omega} = 0.2$ is the oracle~\eqref{eq-def-oracle}, so it achieves zero regret, while the conservative version of our procedure with $\tol=\alpha$ performs quite well for all values of the alternative parameter~$\theta$. In this example, the asymptotic error incurred from conservatively bounding~$\pi_0$ by one in the procedure is small compared to the error incurred by using $\BH(\q)$ at an {\em ad hoc} value. The BH procedure at level~$\frac{\alpha}{\pi_0}$ or~$\alpha$ incurs substantial asymptotic regret by comparison. In particular, note that the $\BH(\alpha/\pi_0)$ procedure incurs the same asymptotic regret as the procedure that rejects nothing; i.e. $\rho(t_{\alpha/\pi_0}^{\BH\text{-POP}}) = \rho(0)$. If we run BH at a lower level like $\alpha/2$, $\alpha/10$, or $\alpha/100$, we can do well for some range of $\theta$ values, but struggle at other parts of the parameter space. No single level for BH dominates in terms of regret, so for the classification risk it is more appropriate to view the BH level as a tuning parameter \edit{than as a proxy for the true~$\lfdr$ threshold~$\alpha$} \citep{neuvial2012false}. 

\subsection{Relationship of our method to the Grenander estimator}\label{sec-grenander}

Since the marginal density $f$ appears in the denominator of the lfdr, bounding $\pi_0\leq 1$ and plugging in Grenander's estimator $\hat{f}_m$ (defined in (\ref{eq-gren})) gives the conservative estimate
\begin{align*}
    \widehat{\lfdr}(t) \coloneqq \frac{1}{\hat{f}_m(t)}, \hspace{1em} t \in [0,1].
\end{align*}
Similar to how the BH procedure chooses an interval $[0,t]$ as large as possible subject to a constraint on an estimate of the FDP, the rejection threshold of the SL procedure can be equivalently expressed as
\begin{align}
\label{argmax-characterization}
    \tau_{\tol}=\underset{p_{(0)},\dots,p_{(m)}}{\operatorname{argmax}} \hspace{.4em} \left\{ \frac{\tol k}{m} - p_{(k)} \right\} = \sup \left\{t \in [0,1] : \widehat{\lfdr}(t) \leq  \tol\right\},
\end{align}
taking the convention that $\sup \emptyset \equiv 0$. The equivalence in~\eqref{argmax-characterization} is illustrated in Figure \ref{fig:Zk_gcm}. Let $\hat{F}_m$ denote the least concave majorant of the empirical cdf $F_m$, plotted as a dashed blue line in the left panel of Figure \ref{fig:Zk_gcm}. By definition of $\widehat{\lfdr}(t)$, the supremum on the right hand side of~\eqref{argmax-characterization} is equal to the largest $t$ for which $\partial_-(\tol\hat{F}_m(t)-t) = \tol\hat{f}_m(t)-1 \geq 0$ (where $\partial_-$ denotes the left derivative), which corresponds to the maximizer of the function $\tol\hat{F}_m(t)-t$, illustrated for example in the right panel of Figure \ref{fig:Zk_gcm}.  $\hat{F}_m \geq F_m$ implies
\begin{align*}
    \tol \hat{F}_m(t)-t \geq \tol F_m(t)-t,\hspace{1em} t\in[0,1],
\end{align*}
with equality at the knots of $\hat{F}_m$, and since the maximizer of the left hand side occurs at a knot of $\hat{F}_m$, it is also the maximizer of the right hand side, i.e. the argmax of $\frac{\tol k}{m}-p_{(k)}$. 
\begin{figure}
	\centering
	\centerline{
	\includegraphics{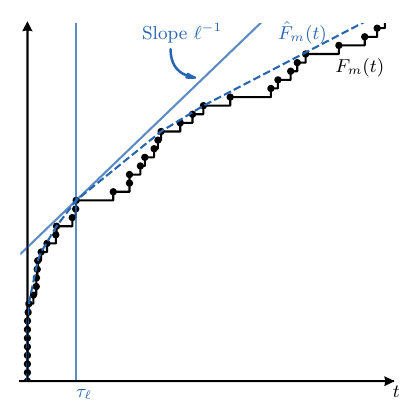}
	\includegraphics{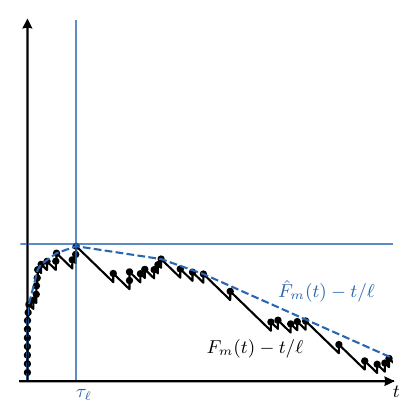}
	} 
    \caption{Left: empirical cdf~$F_m$ and its least concave majorant~$\hat{F}_m$. The support line of slope~$\tol^{-1}$ touches both curves at the decision threshold~$\tau_{\tol}$. Right: same plot with the line~$t/\tol$ subtracted off.}
  \label{fig:Zk_gcm}
\end{figure}

We can again compare this result with the $\BH(\q)$ threshold, given by
\begin{align*}
    \tau_\q^{\BH} =\underset{k=0,\dots,m}{\operatorname{max}} \hspace{.4em} \left\{p_{(k)}: \frac{\q k}{m} - p_{(k)} \geq 0 \right\} = \sup\left\{t \in [0,1] : F_m(t) \geq \q^{-1}t \right\},
\end{align*}
which is the largest $t$ for which the ray $\q^{-1} t$ lies below the ecdf $F_m(t)$. Our procedure instead finds the last intersection of the graph of $F_m$ with a support line of slope $\tol^{-1}$, since
\begin{align*}
    \widehat{\lfdr}(t) \leq \tol \iff \hat{f}_m(t) \geq \tol^{-1}.
\end{align*}
This relationship is illustrated in the left panel of Figure \ref{fig:Zk_gcm}.

\subsection{Asymptotic behavior of our procedure}\label{sec-ao}

Equation \eqref{eq-regret-fixed-t-taylor} suggests that, when $f$ is sufficiently regular near $\tau^*$, the regret is closely related to the squared error of the rejection threshold. Our main result in this section establishes cube-root asymptotics for the behavior of our procedure $\cR_{\tol}$ with $\tol=\alpha/\hat{\pi}_0$, where $\hat{\pi}_0$ consistently estimates $\pi_0$; if $\pi_0$ is known, then the results apply directly with $\hat\pi_0 = \pi_0$. 

We derive limiting distributions for the threshold $\tau_{\tol}$, the $\lfdr$ at the threshold, and the regret of $\cR_{\tol}$. All three are given in terms of Chernoff's distribution \citep{chernoff1964estimation}, which is defined as the distribution of the maximizer~$Z$ of a standard two-sided Brownian motion~$W = (W(t))_{t\in \mathbb{R}}$ with parabolic drift:
\begin{equation}\label{eq-chernoff}
Z = \argmax_{t\in \R}\, W(t) - t^2.
\end{equation}
The random variable~$Z$ has a density with respect to the Lebesgue measure on~$\R$ that is symmetric about zero. \citet{dykstra1999distribution} suggest approximating the density and cdf of $Z$ by those of~$\mathcal{N}\left(0, (.52)^2\right)$. This approximation can be somewhat crude but gives a rough sense for the distribution of~$Z$. \citet{groeneboom2001computing} provide much more accurate numerical methods to compute the density, cdf, quantiles and moments of~$Z$.

\begin{theorem}\label{thm-asymptotics}
Suppose $p_1,\ldots,p_m$ follow the Bayesian two-groups model~\eqref{eq-two-groups}, with  $\pi_0 \in (0,1)$, $f_0 = 1_{[0,1]}$, and $f_1$ non-increasing. For $\tol \in (0, \pi_0^{-1})$, assume additionally that
\begin{enumerate}[(i)]
    \item there is a unique value $t_{\tol}\in (0,1)$ for which $f(t_{\tol}) = \tol^{-1}$, 
    \item $f$ is continuously differentiable in a neighborhood of $t_{\tol}$ with $f'(t_{\tol}) < 0$, and
    \item $\hat\tol$ is any random variable with $m^{1/3}(\hat{\tol} -\tol) \stackrel{p}{\to} 0$ as $m \to \infty$.
\end{enumerate}
Then we have, as $m\to \infty$,
\begin{align}
    \label{eq-threshold-asymptotics}
    m^{1/3}(\tau_{\hat{\tol}} -t_{\tol}) 
    &\;\stackrel{d}{\to}\; \left(\frac{\tol}{4} \cdot f'(t_{\tol})^2\right)^{-1/3} Z,\qquad \text{and}\\[7pt]
   \label{eq-lfdr-asymptotics}
    m^{1/3}\cdot \frac{\lfdr(\tau_{\hat{\tol}}) - \pi_0 \tol}{\pi_0 \tol}
    &\;\stackrel{d}{\to}\; \left(4 \tol^2 \cdot |f'(t_{\tol})|\right)^{1/3} Z.
\end{align}
where $Z$ follows Chernoff's distribution defined in \eqref{eq-chernoff}. Further, suppose that
\begin{equation}\label{eq-stronger}
\P\{m^{-1/3}|\hat\tol - \tol| > \varepsilon\} = o\left(m^{-2/3}\right), \quad \text{ for all } \varepsilon > 0.
\end{equation}
Then we also have $\E \left[\tau_{\hat\tol}\right]
\;\to\; t_{\tol}$. In addition,
\begin{align}
\label{eq-var-tau}
m^{2/3}\textnormal{Var}\left(\tau_{\hat\tol}\right)
&\;\to\; \left(\frac{\tol}{4} \cdot f'(t_{\tol})^2\right)^{-2/3} \textnormal{Var}(Z),\qquad \text{and}\\[7pt]
\label{eq-var-lfdr}
m^{2/3}\textnormal{Var}\left(\frac{\lfdr(\tau_{\hat{\tol}}) - \pi_0\tol}{\pi_0\tol}\right) 
&\;\to\; \left(4\tol^2 \cdot |f'(t_{\tol})|\right)^{2/3}\textnormal{Var}(Z),
\end{align}
where $\textnormal{Var}(Z) \approx 0.26$.
\end{theorem}

The proof of Theorem~\ref{thm-asymptotics} is deferred to the Appendix. It is well-known that the Grenander estimator~$\gren$ estimates~$f$ at a cube root rate pointwise, away from zero, but this result, due to \citep{rao1969estimation}, is too weak to describe the behavior of our procedure. We rely on a stronger version of this result that approximates the local behavior of the Grenander estimator near~$t_{\tol}$. 

The distributional result \eqref{eq-lfdr-asymptotics} complements our result from Theorem~\ref{thm-main}, by showing that $\lfdr(\tau_{\tol}) = \max_{i\in\cR_{\tol}} \lfdr(p_i)$ is not only controlled in expectation, but also concentrates at rate $m^{-1/3}$ around its expectation. In particular, because $\P\{Z \geq 1\} \approx 0.05$, we have
\[
\frac{\lfdr(\tau_{\tol}) - \pi_0\tol}{\pi_0\tol} \;\leq\; m^{-1/3} \left(4\tol^2 \cdot |f'(t_{\tol})|\right)^{1/3},
\]
with roughly $95\%$ probability in large samples. For example, suppose we use $\tol=0.2$, so $f(t_{\tol}) = 5$, and suppose that $f'(t_{\tol})=-50$. Then, whereas Theorem~\ref{thm-main} guarantees $\E\left[\lfdr(\tau_{\tol})\right] \leq 0.2$ exactly, the asymptotic estimate from Theorem~\ref{thm-asymptotics} bounds the $95$th percentile of $\lfdr(\tau_{\tol})$ at $0.24$ if $m=1000$, or at $0.21$ if $m=64,000$.

To understand why the error is of order $m^{-1/3}$, consider fixed $\tol$ and recall that the threshold $\tau_{\tol}$ maximizes the stochastic process 
\[
U(t) \coloneqq F_m(t) - F_m(t_{\tol}) - \frac{t-t_{\tol}}{\tol}.
\]
Because $f(t_{\tol}) = \tol^{-1}$, we have for $t$ near $t_{\tol}$,
\[
F(t) - F(t_{\tol}) \;\approx\; \frac{t-t_{\tol}}{\tol} + \frac{f'(t_{\tol})}{2} (t-t_{\tol})^2.
\]
Introducing the local parameterization $t = t_{\tol} + m^{-a}h$ for $a>0$ leads to
\[
U(t_{\tol} + m^{-a}h) \;\approx\; -\frac{|f'(t_{\tol})|}{2} \cdot \frac{h^2}{m^{2a}} \,+\, \mathcal{N}\left(0, \,\frac{h}{\tol m^{a+1}}\right).
\]
Setting $a=1/3$ balances the mean and variance, giving
\[
m^{2/3} U(t_{\tol} + m^{-1/3}h) \;\stackrel{d}{\to}\; -\frac{|f'(t_{\tol})|}{2} h^2 + \mathcal{N}\left(0, \frac{h}{\tol}\right).
\]
Under this local scaling, $U(t)$ converges to a Brownian motion with parabolic drift, and its maximizer $\tau_{\tol}$ converges to Chernoff's distribution. Theorem~\ref{thm-asymptotics} applies a more careful version of this argument, replacing $F_m(t)$ with its least concave majorant (LCM) $\hat{F}_m(t)$ and applying a result characterizing the process $\gren(t)$ under the same local scaling. The corresponding results for $\lfdr(\tau_{\tol})$ follow from first-order Taylor expansion of $\lfdr(t) = \pi_0/f(t)$ around $t_{\tol}$.

By specializing Theorem~\ref{thm-asymptotics} to $\tol=\alpha/\pi_0$ and $\hat{\tol} = \alpha/\hat\pi_0$, we obtain the limiting regret for our procedure with a known or accurately estimated null proportion.

\begin{theorem}\label{thm-regret}
Suppose $p_1,\ldots,p_m$ follow the Bayesian two-groups model~\eqref{eq-two-groups}, with  $\pi_0 \in (0,1)$, $f_0 = 1_{[0,1]}$, and $f_1$ non-increasing. Assume additionally that
\begin{enumerate}[(i)]
    \item there is a unique value $\tau^*\in (0,1)$ for which $\lfdr(\tau^*) = \frac{\pi_0}{f(\tau^*)} = \alpha$, 
    \item $f$ is continuously differentiable in a neighborhood of $\tau^*$ with $f'(\tau^*) < 0$, and
    \item $\hat\pi_0$ is any estimator of $\pi_0$ with $\P\left\{m^{1/3}(\hat\pi_0 - \pi_0) > \varepsilon\right\}  = o\left(m^{-2/3}\right)$ for all $\varepsilon > 0$.
\end{enumerate}
Then we have, as $m\to \infty$,
\begin{equation}
   \label{eq-regret-asymptotics}
    m^{2/3}\Regret_m(\cR_{\alpha/\hat\pi_0})
    \;\to\; \left(\frac{\alpha^2}{2\pi_0^2} \cdot |f'(\tau^*)|\right)^{-1/3} \textnormal{Var}(Z),
\end{equation}
where $Z$ follows Chernoff's distribution defined in \eqref{eq-chernoff}, and $\textnormal{Var}(Z) \approx 0.26$.
\end{theorem}

Theorems~\ref{thm-asymptotics}--\ref{thm-regret} deal with the regret for $\pi_0 \in (0,1)$. Under the global null, represented in the Bayesian model by $\pi_0 = 1$, the behavior is different and the regret is simply $\omega\E V$, which is $O(m^{-1})$, as we see next.

\begin{proposition}\label{prop-global-null} Suppose~$(p_i)_{i=1}^m$ follow a two-groups model~\eqref{eq-two-groups} with $f_0 = 1_{[0,1]}$ and~$\pi_0=1$, i.e. $H_i=0$ for all~$i$ and~$p_i\simiid\textnormal{Unif}(0,1)$. Then as $m\to \infty$, we have
\[
m\,\Regret_m(\cR_{\tol})
\to \omega \sum_{k=1}^\infty\P\left\{U_k \le\tol\right\}, \qquad \textnormal{for } U_k\sim \textnormal{Gamma}(k, k),
\]
which is finite for every~$\tol\in [0, 1)$.
\end{proposition}
Proposition~\ref{prop-global-null} is closely related to results derived in \citet{finner2002multiple}.

\section{Numerical results}\label{sec-simulations}

\subsection{Demonstration of theoretical results}

This section highlights our main results on simulation experiments. We adapt a simulation setting of~\citet{benjamini1995controlling} to the two-groups model~\eqref{eq-two-groups}. The observations are~$m=64$ independent, normally distributed random variables~$Y\sim\mathcal{N}(\mu, I_m)$, and the~$i^\textnormal{th}$ null hypothesis is that~$\mu_i = 0$, i.e. $H_i\coloneqq 1\{\mu_i\ne 0\}$. The component means~$\mu_i$ are independent and identically distributed random variables with
\begin{align}\label{eq-means-sim}
\mu_i\simiid \begin{cases}
0 & \textnormal{with probability } \frac{3}{4}, \\
5\frac{j}{4} & \textnormal{with probability } \frac{1}{16}, \textnormal{ for } j=1,\ldots,4.
\end{cases}
\end{align}
We compute one-tailed~$p$-values $p_i = \bar{\Phi}(Y_i)$, where $\bar{\Phi}$ denotes the standard Gaussian survival function. The pairs~$(H_i, p_i)_{i=1}^m$ follow a two-groups model with~$\pi_0 = 0.75$, $f_0 = 1_{[0,1]}$ and alternative density
\begin{align}\label{eq-f1-bh-E}
f_1(t) = \frac{\frac{1}{4}\sum_{j=1}^4\phi\left(\bar\Phi^{-1}\left(t\right)-5\frac{j}{4}\right)}{\phi\left(\bar\Phi^{-1}\left(t\right)\right)} \qquad \textnormal{for } 0\le t\le 1,
\end{align}
where $\phi$ denotes the probability density function of the standard Gaussian distribution. The top half of Figure~\ref{fig-main-sim} shows the mixture density and corresponding~$\lfdr$. 

We repeatedly sample from the above two-groups model for a total of~$10^5$ simulation replicates. The bottom half of Figure~\ref{fig-main-sim} shows the~$\FDR$ (left) and $\maxlfdr$ (right) for our procedure, at conservative level $\tol$ and estimated level $\tol/\hat\pi_0^\lambda$ with threshold~$\lambda=\frac{1}{2}$\edit{, along with the corresponding quantities for the BH procedure at levels $\q$ and $\q/\hat\pi_0^\lambda$ shown for comparison.} The BH procedure at level~$\q$, shown as a solid red line, achieves~$\FDR$ exactly~$\pi_0\q$, whereas its $\maxlfdr$ can be much larger than $\pi_0\q$. For instance, the BH procedure at level~$\q = 0.2$ has~$\maxlfdr$ above~$50\%$, so the least promising rejection is more likely to be null than non-null. By contrast, the SL procedure, shown as a solid blue line, controls $\FDR$ substantially below the level~$\pi_0\tol$ but has $\maxlfdr$ equal to~$\pi_0\tol$ as guaranteed by Theorem~\ref{thm-main}. The modified BH and SL procedures that incorporate~$\hat\pi_0^\lambda$ achieve~$\FDR$ and~$\maxlfdr$ just below~$\q$ and~$\tol$, respectively. \edit{Importantly, because the BH procedure at level $q$ is not intended to control $\maxlfdr$ at level $q$, these results do not indicate a failure of the BH procedure to achieve its advertised control; likewise, because the SL procedure at level $\tol$ is not intended to target FDR control at level $\tol$, the fact that its FDR is well below $\pi_0\tol$ does not indicate that the method is overly conservative.}

\edit{In Figure~\ref{fig-lfdr-asymptotics}, we assess how well the maximum $\max_{i\in \cR}\lfdr(p_i)$ concentrates around its expectation $\maxlfdr(\cR)$ by plotting the interquartile range of $\max_{i\in \cR}\lfdr(p_i)$ for the BH and SL procedures. The blue x's indicate the asymptotic prediction~\eqref{eq-lfdr-asymptotics} of Theorem~\ref{thm-asymptotics}. For $m=1024$, the maximum concentrates well, and the theoretical prediction is quite accurate.}

\FloatBarrier

\begin{figure}[p!]
	\centering
    \centerline{
	\includegraphics{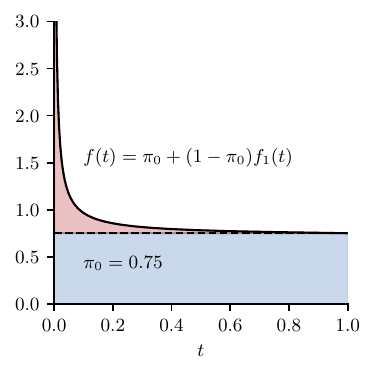}
	\includegraphics{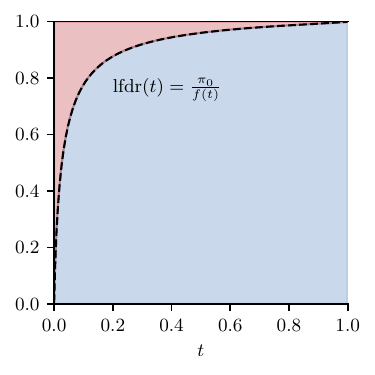}
	}
	\caption{Simulation for $m=64$ independent, one-tailed $Z$-tests. Above: Mixture density~$f$ (left) and~$\lfdr$ (right), with alternative density~$f_1$ defined in~\eqref{eq-f1-bh-E} and null proportion~$\pi_0 = 0.75$. Note~$f_1$ diverges as $t\downarrow 0$. Below: Comparison of $\FDR$ control (left) and~$\maxlfdr$ control (right) on simulated data. The estimate of the null proportion is~\eqref{eq-def-storey} with~$\lambda=0.5$. The black, dash-dotted lines have slopes~$1$ and~$0.75$.}\label{fig-main-sim}
	\centerline{
	\includegraphics{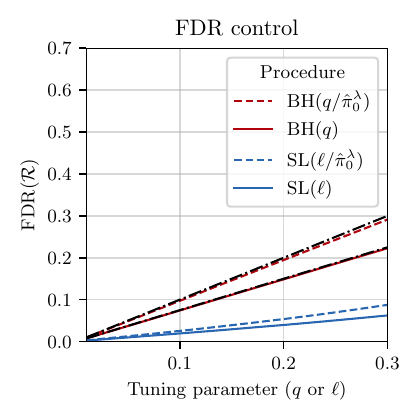}
	\includegraphics{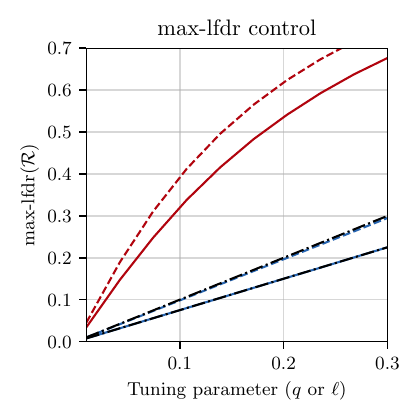}
	}
\end{figure}

\begin{figure}[p!]
    \centering
	\centerline{
	\includegraphics{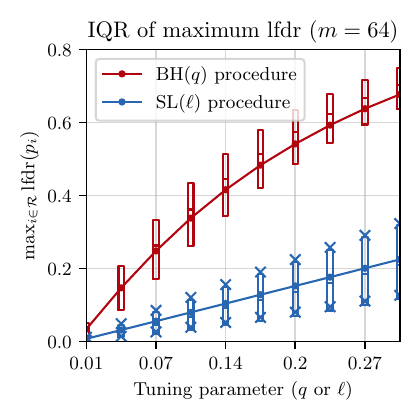}
	\includegraphics{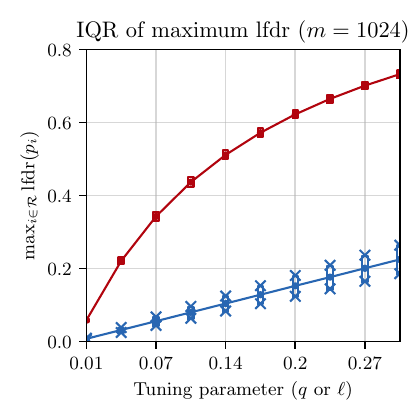}
	}
	\caption{Interquartile range of~$\max_{i\in \mathcal{R}}\lfdr(p_i)$ for the SL and BH procedures as a function of the input levels~$\tol$ and $\q$, respectively. The blue x's indicate the asymptotic predictions of the IQR from Theorem~\ref{thm-asymptotics}. For this simulation, we used the alternative density~$f_1$ defined in~\eqref{eq-f1-bh-E} and null proportion~$\pi_0 = 0.75$. Left: $m=64$ hypotheses. Right: $m=1,024$ hypotheses.}\label{fig-lfdr-asymptotics}
 ~\\~\\
	\centering
	\centerline{
	\includegraphics[height=20em]{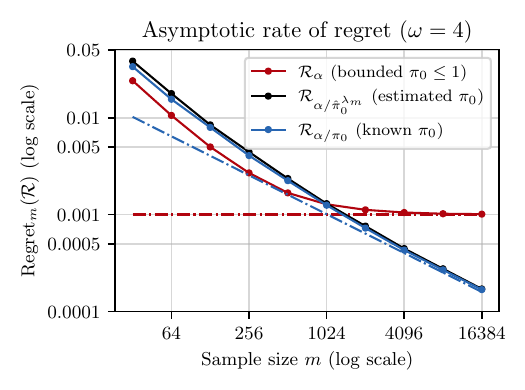}
	}
	\caption{A log-log plot of the regret~\eqref{eq-def-regret} as a function of the sample size~$m$. The blue dash-dotted line shows the asymptotic prediction~\eqref{eq-regret-asymptotics} of Theorem~\ref{thm-regret}, and the red dash-dotted line shows the asymptotic regret~\eqref{eq-regret-fixed-t} of the inconsistent procedure~$\mathcal{R}_\alpha$ which bounds the null proportion~$\pi_0\le 1$ instead of estimating it. For this simulation, we used the alternative density~$f_1$ defined in~\eqref{eq-f1-bh-E}, cost ratio~$\omega=4$ and null proportion~$\pi_0 = 0.75$.
	}\label{fig-regret-asymptotics}	
\end{figure}

\FloatBarrier

Figure~\ref{fig-regret-asymptotics} shows a log-log plot of the regret as a function of the sample size~$m$. The red curve shows the regret of our uncorrected procedure~$\cR_\alpha$ for~$\alpha = 0.05$, which asymptotically tends to $\rho(t_{\alpha})$ and hence asymptotically incurs some non-vanishing regret \edit{(shown as a red, dash-dotted line)} described in Section~\ref{sec-popn-regret}. The \edit{solid} blue curve shows the regret of the corrected procedure~$\cR_{\alpha/\pi_0}$ with known~$\pi_0$. For larger samples, the simulated regret closely matches the asymptotic prediction from~\eqref{eq-regret-asymptotics}, shown \edit{as a dash-dotted blue line}. The \edit{solid black} curve (which is nearly indistinguishable from the \edit{solid blue} curve) shows the corrected procedure with an estimated null proportion~$\hat{\pi}_0^\lambda$ based on~\eqref{eq-def-storey} with~$\lambda = 1-m^{-1/5}$. 

\subsection{Robustness of results to assumptions}\label{sec-robustness}

\edit{Next, we assess the robustness of~$\maxlfdr$ control to certain violations of the assumptions of Theorem~$\ref{thm-main}$, starting with the independence assumption. When reasoning about the behavior of our procedure under dependence, we ought to keep in mind that the $p$-value threshold~$\tau_\tol$ performs a simple operation on the ecdf:
\[
\tau_\tol = \argmin_{t\in [0,1]}\,t-\tol F_m(t),
\]
where we take the largest minimizer to agree with definition~\eqref{eq-def-procedure}. Roughly, if the assumptions are violated in a way that does not have a major influence on the fluctuations of $F_m$ around its expectation $F = \mathbb{E}F_m$, we can expect that, for large~$m$ and under regularity conditions, the SL threshold~$\tau_\tol$ approaches the population minimizer
\begin{align}\label{eq-popn-threshold} 
t_\tol = \argmin_{t\in [0,1]}\,t -\tol F(t).    
\end{align}

\begin{proposition}\label{prop-dependence} Suppose $p_1,\ldots,p_m\in [0,1]$ are (possibly dependent) continuous random variables. Let $F_m$ denote the empirical cdf of $p_1,\ldots,p_m$, and let $F = \mathbb{E}F_m$ denote the average marginal cdf. Assume $F$ is concave and that there is a unique minimizer~$t_\tol$ in~\eqref{eq-popn-threshold}. If $|F_m-F|_\infty\stackrel{p}{\to} 0$, then $\tau_{\tol} \stackrel{p}{\to}t_{\tol}$.
\end{proposition}

Proposition~\ref{prop-dependence} shows that a sufficient condition for consistency of the SL threshold~$\tau_{\tol}$ to the population threshold~$t_{\tol}$ is uniform convergence of the ecdf $|F_m-F|_\infty\stackrel{p}{\to}0$. If $F$ is also continuously differentiable, then we may further conclude $\frac{\pi_0}{f(\tau_\tol)}\stackrel{p}{\to}\frac{\pi_0}{f(t_\tol)}$, where $t_\tol$ is the largest $t\in  [0,1]$ such that $\frac{1}{f(t)}\le\tol$. In particular, if the $p$-values are identically distributed, this result implies that the $\lfdr$ is controlled asymptotically, i.e. $\lfdr(\tau_\tol) \stackrel{p}{\to}\lfdr(t_\tol) \le \pi_0\tol$.

We next illustrate Proposition~\ref{prop-dependence} by considering two simulation settings in which the observations $Y$ are dependent, one in which uniform convergence holds and one in which it fails. Let}~$Y\sim \mathcal{N}(\mu, \Sigma)$ for some~$m\times m$ positive definite covariance matrix~$\Sigma$, where the means~$\mu_i$ are iid as before according to~\eqref{eq-means-sim}. We consider the equicorrelation model
\begin{align}\label{eq-equicorrelated-covariance}
\Sigma_{\textnormal{EQ}} 
\coloneqq
\begin{bmatrix}
1 & \rho & \cdots & \cdots & \rho \\
\rho & 1 & \ddots & & \vdots \\
\vdots & \ddots & \ddots & \ddots & \vdots \\
\vdots&&\ddots&\ddots& \rho \\
\rho&\cdots&\cdots&\rho & 1 \\
\end{bmatrix}
\end{align}
for some correlation~$\rho\in (-\frac{1}{m-1}, 1)$. We also consider a stationary autoregressive model
\begin{align}\label{eq-autoregressive-covariance}
\Sigma_{\textnormal{AR}} \coloneqq \begin{bmatrix}
1 & \rho & \rho^2 &  & \cdots & \rho^{m-1}  \\
\rho & 1 & \rho & \rho^2 &  & \rho^{m-2} \\
\rho^2 & \rho & 1 & \rho &\ddots &  \rho^{m-3} \\
\vdots & \ddots & \ddots & \ddots & \ddots &\vdots \\
\vdots&&\ddots&\ddots& 1 &\rho \\
\rho^{m-1} & \cdots & \cdots & \rho^2 & \rho & 1
\end{bmatrix}
\end{align}
for any autocorrelation~$\rho$ satisfying $|\rho| < 1$. 

Figure~\ref{fig-dependence} shows the results of our simulation under dependence as a function of the marginal correlation~$\rho$, fixing the target level at~$\tol=0.2$. The SL procedure has~$\maxlfdr$ above~$\pi_0\tol$ in both cases, with~$\maxlfdr$ increasing with the correlation~$\rho$. The modified SL procedure inflates~$\maxlfdr$ subtantially above level~$\tol$, especially in the equicorrelated model. \edit{These results corroborate Proposition~\ref{prop-dependence}, since uniform convergence of the ecdf fails in the equicorrelated model but holds in the autoregressive model \citep{tucker1959generalization}.} 

Minimizing weighted classification risk under dependence requires thresholding the local false discovery rate using the full posterior, i.e. $\mathbb{P}\left(H_i = 0 \mid p_1,\ldots, p_m\right).$ However, the oracle~$\mathcal{R}^*$ (see~\eqref{eq-def-oracle}) rejects~$p$-values on the basis of~$\lfdr(t) = \mathbb{P}\left(H_i = 0 \mid p_i=t\right)$, the posterior probability of the null given only the corresponding~$p$-value. Hence, even if we could consistently estimate~$\pi_0$ and~$F$, our procedure~$\mathcal{R}_{\alpha/\hat{\pi}_0}$ would only target the best~\emph{separable} oracle~$\mathcal{R}^*$, and there may be a considerable gap between the risk of the best separable rule and the risk of the full Bayesian analysis. 

We also simulate a setting in which the alternative density~$f_1$ is not monotone. \edit{For this simulation, we set $\pi_0 = 0.75$, and we let the non-null $p$-values follow an equal mixture of Beta$(2, 100)$ and Beta$(.01, 2)$ distributions: the alternative density is
\begin{align}\label{eq-misspecified-alternative}
f_1(t)
\coloneqq \frac{1}{2}\frac{t(1-t)^{99}}{\text{B}(2, 100)} + \frac{1}{2}\frac{t^{-.99}(1-t)}{\text{B}(.01, 2)},    
\end{align}
where $\text{B}(\alpha, \beta)$ denotes the Beta function.} Figure~\ref{fig-misspecification} shows the results of this simulation. Since the non-decreasing~$\lfdr$ assumption is violated, our procedure does not control~$\maxlfdr$. \edit{Of course, the first part of Theorem~\ref{thm-main} still applies, so the SL procedure controls the~$\lfdr$ at the decision boundary~$\tau_\ell$ at level~$\pi_0\ell$; however, as the top right of Figure~\ref{fig-misspecification} indicates, the~$\max_{i\in \cR_\tol} \lfdr(p_i)$ can be strictly larger than~$\lfdr(\tau_\tol)$, because the maximum lfdr may be attained in the rejection region's interior.} By contrast, the BH procedure still controls~$\FDR$ exactly at level~$\pi_0\tol$.  

\FloatBarrier

\begin{figure}[p!]
	\centering
    \centerline{
	\includegraphics{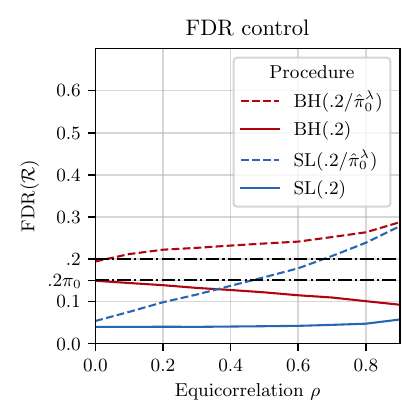}
	\includegraphics{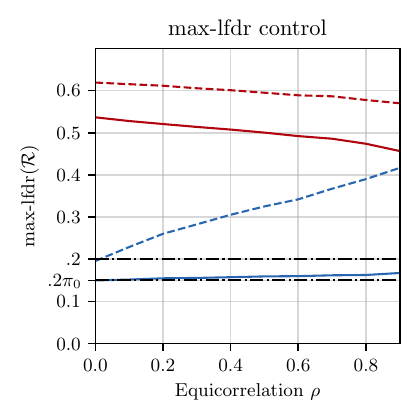}
	}
	\caption{Simulations for dependent, one-tailed $Z$-tests with nominal level~$\tol = \q = 0.2$, null proportion~$\pi_0 = 0.75,$ and $m=64$ hypotheses. The BH procedure is shown in red and the SL procedure in blue. The solid lines indicate the results for the procedures run at input level~$\tol$, and the dashed lines indicate the results for the modified procedures based on estimating~$\pi_0$ with~$\hat{\pi}_0^\lambda$ where~$\lambda=0.5$. Above: results for equicorrelated model~\eqref{eq-equicorrelated-covariance}. Below: results for autoregressive model~\eqref{eq-autoregressive-covariance}.}\label{fig-dependence}
	\centerline{
	\includegraphics{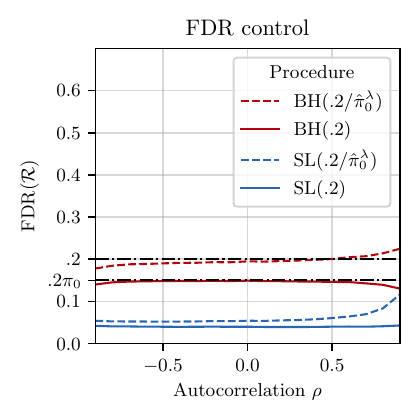}
	\includegraphics{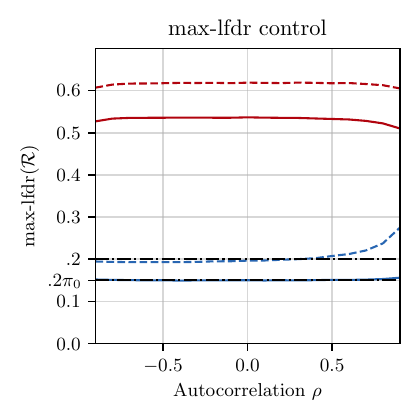}
	}
\end{figure}

\begin{figure}[p!]
	\centering
    \centerline{
    \includegraphics{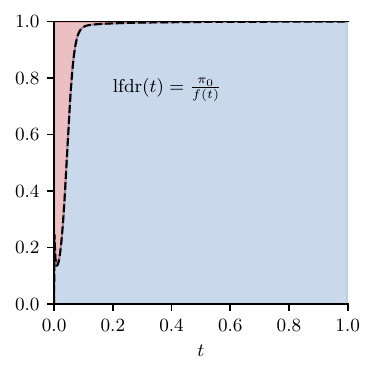}
	\includegraphics{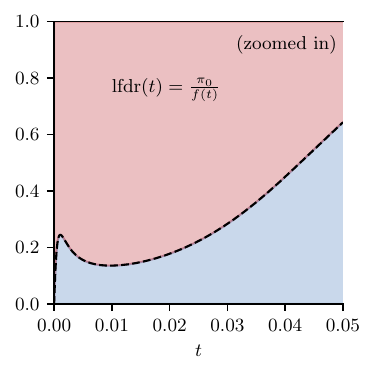}
	}
	\caption{Simulation for \edit{independent tests with misspecified shape constraint on the alternative density~\eqref{eq-misspecified-alternative}, null proportion~$\pi_0=0.75$, and $m=64$ hypotheses}. Above: \edit{$\lfdr$ shown over its entire domain $[0,1]$ (left) and zoomed into the interval $[0, 0.05]$ (right) to emphasize the violation in monotonicity.} Below: Comparison of $\FDR$ control (left) and~$\maxlfdr$ control (right) on simulated data. The estimate of the null proportion is~\eqref{eq-def-storey} with~$\lambda=0.5$. The black, dash-dotted lines have slopes~$1$ and~$0.75$.}\label{fig-misspecification}
	\centerline{
	\includegraphics{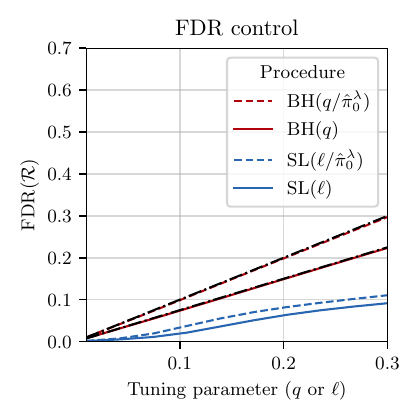}
	\includegraphics{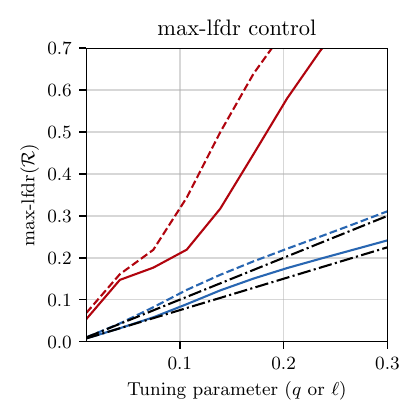}
	}
\end{figure}

\section{Discussion}\label{sec-discussion}

In this work we have introduced a new error criterion, the $\maxlfdr$, which modifies the $\FDR$ by redirecting attention away from the average quality of the rejection set and toward the rejections that are close to the rejection boundary. Despite the seeming difficulty of measuring the quality of a single rejection, we also introduce a simple new multiple testing procedure that controls the $\maxlfdr$ at level $\pi_0\tol$ in finite samples, where $\tol$ is a tuning parameter and $\pi_0$ is the null proportion. We assume only that the data follow a Bayesian two-groups model in which smaller $p$-values reflect stronger evidence against the null. We find that our method is better able than the BH method to adapt to the unknown problem structure, and to perform well without knowledge of the true underlying distribution.

The BH procedure owes its enduring utility for $\FDR$ control in part to its versatility beyond this basic setting, however. It is known to still control $\FDR$, for instance, when the null $p$-values are super-uniform and under certain forms of positive dependence, two of many possible extensions that we leave open for our procedure.

Another seeming advantage of the $\FDR$ criterion is that it requires no Bayesian assumptions, whereas the $\maxlfdr$ is only defined with reference to a Bayesian model. \edit{In particular, in the two-groups model we consider, the pairs~$(H_i, p_i)_{i=1}^m$ are independent and identically distributed. We might prefer to treat the sequence~$(H_i)_{i=1}^m$ as fixed, and $(p_i)_{i=1}^m$ as independent but {\sl{not necessarily}} identically distributed.} A possible avenue for generalizing the $\maxlfdr$ to frequentist settings is to work with its characterization as the probability that the last rejection is a false discovery. Indeed, our proof of Theorem~\ref{thm-main} implies that \edit{this probability} is controlled even conditional on~$(H_i)_{i=1}^m$. This is initially puzzling: if each $H_i$ is fixed, then how can we speak of the probability that the last rejection is a false discovery? The answer is that $H_{(R)}$ is random even if~$(H_i)_{i=1}^m$ is fixed, since the index~$R$ is random. We leave further development of the frequentist connection to the $\maxlfdr$ to future work.

\FloatBarrier

\begin{appendix}
\section{Proofs}

\subsection{Estimation of \texorpdfstring{$\pi_0$}{the null proportion}} We first prove that our modified procedure~\eqref{eq-def-procedure-modification} controls the~$\maxlfdr$ at level at most~$\tol$.

\begin{proof}[Proof of Theorem~\ref{thm-main-modification}.]
As in the proof of Theorem~\ref{thm-main}, we have 
\[
\maxlfdr\left(\cR_{\tol}^\lambda\right) \;=\; \P\left\{H_{(R_{\tol}^\lambda)} = 0, \,R_{\tol}^\lambda > 0\right\} \;=\; m\,\P\left\{H_m = 0,\, p_{(R_{\tol}^\lambda)} = p_m\right\}.
\]
Define the $\sigma$-field $\mathcal{F} = \sigma\left(p_1,\ldots,p_{m-1},H_m, 1\{p_m \leq \lambda\}\right)$. We restrict our attention to the event $A \coloneqq \{H_m = 0, p_m \leq \lambda\}$, since the event $\{H_m = 0,\, p_{(R_{\tol}^\lambda)} = p_m\}$ cannot occur except on $A$. On $A$, which is $\mathcal{F}$-measurable, we have $p_m/\lambda \mid \mathcal{F} \sim U[0,1]$.

Let $m^\lambda = \#\{i:\; p_i \leq \lambda\}$, which is also $\mathcal{F}$-measurable. If $j_1 \leq \cdots \leq j_{m^{\lambda}} = m$ are the indices of the $p$-values that are below $\lambda$, define the modified $p$-values $p_i^\lambda = p_{j_i}/\lambda$, for $i=1,\ldots,m^\lambda$. Because the order statistics of $\lambda p_1^\lambda, \ldots, \lambda p_{m^\lambda}^\lambda$ are also the first $m^\lambda$ order statistics of $p_1,\ldots,p_m$, the quantity $R_{\tol}^\lambda$ defined in \eqref{eq-def-procedure-modification} can be rewritten as
\begin{align*}
  R_{\tol}^\lambda &\;=\; \argmin_{k = 0, \ldots, m^\lambda} \lambda p_{(k)}^\lambda -  \frac{\tol}{\hat\pi_0^\lambda}\cdot \frac{k}{m}\\[5pt]
  &\;=\; \argmin_{k = 0, \ldots, m^\lambda} p_{(k)}^\lambda - \frac{\tol^\lambda k}{m^\lambda}, \quad \text{ for }\tol^\lambda = \frac{\tol m^\lambda}{\lambda \hat\pi_0^\lambda m}.
\end{align*}
Applying Lemma~\ref{lem-last-rejection}, we have
\[
\P\left\{H_m = 0, \,p_{(R_{\tol}^\lambda)} = p_m \mid \mathcal{F}\right\} 
\le \frac{\tol^\lambda}{m^\lambda} \,\cdot\, 1_A = \frac{\tol}{\lambda \hat\pi_0^\lambda m} \,\cdot\, 1_A.
\]
Marginalizing over $\mathcal{F}$, and noting that $\P(A) = \pi_0 \lambda$, we obtain
\begin{align*}
\P\left\{H_m = 0, \,p_{(R_{\tol}^\lambda)} = p_m \right\} 
&\;\le\; \frac{\tol}{m} \cdot \E\left[\frac{\pi_0}{\hat\pi_0^{\lambda}} \mid A\right] \\[5pt]
&\;=\; \frac{\tol}{m} \cdot \frac{(1-\lambda)\pi_0}{1-F(\lambda)} \cdot \E\left[\frac{(1-F(\lambda)) m}{1 + \#\{i < m:\; p_i > \lambda\}}\right]\\[5pt]
&\;=\; \frac{\tol}{m} \cdot \frac{(1-\lambda)\pi_0}{1-F(\lambda)} \cdot (1-F(\lambda)^m)\\
&\;\leq\; \frac{\tol}{m},
\end{align*}
completing the proof. The final inequality is a standard binomial identity:
\begin{align*}
\E\left[ \frac{\beta m}{1+\text{Binom}(m-1,\beta)} \right]
&\;=\; \sum_{k=0}^{m-1} \frac{\beta m}{1+k}\binom{m-1}{k}\beta^k (1-\beta)^{m-1-k}\\[5pt]
&\;=\; \sum_{k=0}^{m-1} \binom{m}{k+1} \beta^{k+1} (1-\beta)^{m-(k+1)}\\[5pt]
&\;=\; \sum_{j=1}^{m} \binom{m}{j} \beta^{j} (1-\beta)^{m-j}\\[5pt]
&\;=\; \P\{\text{Binom}(m,\beta) \geq 1\} \\[5pt]
&\;=\; 1 - (1-\beta)^m,
\end{align*}
completing the proof.
\end{proof}

\subsection{Asymptotics} We begin with a technical lemma that allows us to account for the estimation~$\hat\tol$ of the target level~$\tol$.

\begin{lemma}
\label{lemma-uniform-integrability}
Suppose $p_i \stackrel{\iid}{\sim} f$ for $i=1,\dots,m$ and let $\tau_{\hat{\tol}}$ be the threshold obtained by running our procedure (\ref{eq-def-procedure}) at level $\hat{\tol}$. Also suppose that the first two conditions (i) and (ii) in Theorem \ref{thm-asymptotics} hold, and that the strengthened condition (\ref{eq-stronger}) holds. Then there exists a positive sequence $\varepsilon_m \to 0$ for which
\begin{align}
\label{seq-of-eps}
    \P\{m^{1/3}|\hat{\tol}-\tol| > \varepsilon_m\} = o(m^{-2/3}).
\end{align}
Given such a sequence $(\varepsilon_m)$, define the truncated random variable,
\begin{align*}
Y_m = m^{1/3}(\tau_{\hat{\tol}}-t_{\tol}) \cdot 1_{A_m},\hspace{1em} A_m \coloneqq \left\{|\hat{\tol}-\tol|\leq m^{-1/3}\varepsilon_m, |\tau_{\hat{\tol}}-t_{\tol}| \leq m^{-2/9} \right\},
\end{align*}
Then the sequence $\left\{Y_m^2\right\}$ is uniformly integrable, i.e. for any $\delta>0$ there exists an $M>0$ so large that
\begin{align*}
\sup_{m \in \mathbb{N}} \E \left( Y_m^2 \cdot 1_{\{Y_m^2 > M\}} \right) < \delta.
\end{align*}
\end{lemma}

\begin{proof}
First we show that the condition (\ref{eq-stronger}) implies the existence of a sequence $(\varepsilon_m)$ for which (\ref{seq-of-eps}) holds. Since
\begin{align*}
    m^{2/3}\P\{m^{1/3}(\hat{\tol}-\tol) > \varepsilon\} \to 0 \hspace{1em} \text{ for any }\varepsilon > 0,
\end{align*}
there exists some $m_k \in \mathbb{N}$ for which $m \geq m_k$ implies
\begin{align*}
    m^{2/3} \P\{m^{1/3}(\hat{\tol}-\tol)>1/k\}\leq 1/k,
\end{align*}
and the above statement holds for every $k \in \mathbb{N}$. Further, since the $(m_k)$ can be chosen such that $m_1 < m_2 < \dots$, let $\varepsilon_m \coloneqq 1/k$ for all $m \in (m_{k},m_{k+1}]$. Then $(\varepsilon_m)$ satisfies the property (\ref{seq-of-eps}).

Next we show the uniform integrability condition holds. Since the integrand is non-negative, its mean is the integral of the tail probability,
\begin{align}
\nonumber
\E  \left( Y_m^2 \cdot 1_{\{Y_m^2 > M\}} \right) &= \int_0^\infty \P\{Y_m^2 \cdot 1_{\{Y_m^2 > M\}} > t\} \textnormal{d} t \\
\nonumber
&= \int_0^\infty \P\{|Y_m| > \sqrt{t\vee M}\} \textnormal{d} t \\
\nonumber
&= \int_0^M \P\{|Y_m|>\sqrt{M}\} \textnormal{d} t+ \int_{\sqrt{M}}^\infty \P\{|Y_m| > h\} \cdot 2h\textnormal{d} h \\
\label{tail-probs}
&= M \P\{|Y_m|>\sqrt{M}\} + \int_{\sqrt{M}}^\infty \P\{|Y_m| > h\} \cdot 2h\textnormal{d} h .
\end{align}
For fixed $h>0$, we eventually have $hm^{-1/3}<m^{-2/9}$ and the tail probability can be written
\begin{align*}
\P\{|Y_m| > h\} &\leq  \P\{|\tau_{\hat{\tol}}-t_{\tol}| \geq hm^{-1/3},A_m\} \\
&\leq \P(hm^{-1/3}\leq \tau_{\hat{\tol}}-t_{\tol} \leq m^{-2/9},A_m) + \P(hm^{-1/3}\leq t_{\tol}-\tau_{\hat{\tol}} \leq m^{-2/9},A_m). 
\end{align*}
We just analyze the first piece because an analogous argument yields the same bound for the second one. By definition of $\tau_{\hat{\tol}}$, if $t_{\tol} + hm^{-1/3} \leq \tau_{\hat{\tol}} \leq t_{\tol} + m^{-2/9}$, then we must have
\begin{align*}
\sup_{t \in (t_{\tol} + hm^{-1/3},t_{\tol}+m^{-2/9})} F_m(t)-F_m(t_{\tol}) - {\hat{\tol}}^{-1}(t-t_{\tol}) > 0.
\end{align*}
Let $U_i \coloneqq F(p_i)$ so that $U_i \stackrel{\iid}{\sim} \text{Uniform}(0,1)$. Then the above is equivalent to
\begin{align}
\label{sup-event}
\sup_{u \in B_m} G_m(u)-G_m(F(t_{\tol})) - \hat{\tol}^{-1}(F^{-1}(u)-t_{\tol}) > 0,
\end{align}
where $G_m(u) = \frac{1}{m}\sum_{i=1}^m 1_{\{U_i \leq u\}}$ and $B_m \coloneqq [F(t_{\tol} + hm^{-1/3}),F(t_{\tol}+m^{-2/9})]$. Taylor expanding $F^{-1}(u)$ around $F(t_{\tol})$, we have
\begin{align*}
F^{-1}(u) = t_{\tol} + \tol (u-F(t_{\tol})) + \frac{(F^{-1})''(\xi)}{2} (u-F(t_{\tol}))^2
\end{align*}
for some $\xi \in (F(t_{\tol}), F(t_{\tol} + m^{-2/9}))$. Plugging this in for $F^{-1}(u)$, (\ref{sup-event}) is equivalent to
\begin{align*}
\sup_{u \in B_m} G_m(u)-G_m(F(t_{\tol})) - \hat{\tol}^{-1}\tol \cdot (u-F(t_{\tol})) + \hat{\tol}^{-1} \cdot \frac{f'(F^{-1}(\xi))}{2f(F^{-1}(\xi))^2} \cdot (u-F(t_{\tol}))^2 > 0.
\end{align*}
Let $B_{m,k}\coloneqq [F(t_{\tol}+hkm^{-1/3}),F(t_{\tol}+h(k+1)m^{-1/3})]$ for $k=1,\dots,\lceil h^{-1} m^{1/9} \rceil$. Then $B_m \subset \cup_{k=1}^{\lceil h^{-1} m^{1/9} \rceil} B_{m,k}$, and $u \in B_{m,k}$ implies that for large enough $m$,
\begin{align*}
u &\geq F(t_{\tol}+hkm^{-1/3}) \\
&\geq F(t_{\tol}) + f(t_{\tol}+hkm^{-1/3}) hkm^{-1/3} \tag{MVT} \\
&\geq F(t_{\tol})+ \frac{1}{2\tol}\cdot hkm^{-1/3} \tag{$hkm^{-1/3}\to 0$}.
\end{align*}
Now since $f'(t_{\tol})<0$ and $(u-F(t_{\tol}))^2 \geq (\frac{hkm^{-1/3}}{2\tol})^2$, it suffices to bound the probability of the intersection between $A_m$ and the event
\begin{align*}
\sup_{u \in B_m} G_m(u)-G_m(F(t_{\tol})) - \hat{\tol}^{-1}\tol \cdot (u-F(t_{\tol})) > \hat{\tol}^{-1} \cdot \frac{|f'(t_{\tol})|\left(\frac{hkm^{-1/3}}{2\tol}\right)^2}{2\tol^{-2}},
\end{align*}
where we have used that $f$ is continuously differentiable at $t_{\tol}$. By a union bound, the probability of the intersection between the above event and $A_m$ is bounded by
\begin{equation}
\begin{aligned}
\label{union-bd}
&\leq \sum_{k=1}^{\lceil h^{-1} m^{1/9} \rceil} \P\bigg\{ A_m\cap \bigg\{ \sup_{u \in B_{m,k}} G_m(u) - G_m(F(t_{\tol})) -\\
&~~~~~~~~~~~~~~~~~~~~~~~~~~~~~~~~~~~~~~~\hat{\tol}^{-1}\tol\cdot (u-F(t_{\tol})) > \hat{\tol}^{-1}\cdot \frac{|f'(t_{\tol})| (hkm^{-1/3})^2}{8}\bigg\}  \bigg\} .
\end{aligned}
\end{equation}
Note that the proportion of the $\{ U_i \}_{i=1}^n$ in the interval $[F(t_{\tol}),u]$ is equal in distribution to the proportion of the $\{ U_i \}_{i=1}^n$ in the interval $[0,u-F(t_{\tol})]$. Together with the assumptions that $\hat{\tol}^{-1}\tol = 1+o(m^{-1/3})$ on $A_m$ and $u-F(t_{\tol}) \leq \tol^{-1}h(k+1)m^{-1/3}$ for $u \in B_{m,k}$, for $m$ larger than some constant, the $k^{\textnormal{th}}$ summand is bounded by
\begin{align*}
\leq \P\left\{ \left\{\sup_{u \in B_{m,k}-F(t_{\tol})} G_m(u) - u > \frac{|f'(t_{\tol})| h^2k^2m^{-2/3}}{16\hat{\tol}}\right\} \cap A_m \right\},
\end{align*}
where $B_{m,k}-F(t_{\tol})$ is an interval $(a,b)$ with shifted endpoints
\begin{align*}
a &\coloneqq F(t_{\tol}+hkm^{-1/3})-F(t_{\tol}) \\
b &\coloneqq F(t_{\tol}+h(k+1)m^{-1/3})-F(t_{\tol}).
\end{align*}
Since $F$ is concave, it is below its linearization at $t_{\tol}$, i.e. $b \leq \tol^{-1}h(k+1)m^{-1/3} =: b' $, so the probability is bounded by
\begin{align}
\label{uniform_thing}
\leq \P\left( \left\{ \sup_{u \in (0,b')} G_m(u)- u > \frac{|f'(t_{\tol})| h^2k^2m^{-2/3}}{16\hat{\tol}} \right\} \cap A_m \right).
\end{align}
Now let $N \coloneqq mF_m(b')$ be the number of observations below $b'$. Since
\begin{align*}
G_m(u)-u = \frac{m b'}{N}\cdot G_m(u)  \left( \frac{N}{m b'} - 1\right) + b' \left(\frac{m}{N}\cdot G_m(u) -\frac{u}{b'} \right) ,
\end{align*}
the tower property and the triangle inequality give
\begin{align}
\label{piece1}
(\ref{uniform_thing}) &\leq \E \bigg[ \P\bigg( \left\{ \sup_{u \in (0,b')}\frac{m b'}{N}\cdot G_m(u)  \left| \frac{N}{m b'} - 1\right|  > \frac{|f'(t_{\tol})| h^2k^2m^{-2/3}}{32\hat{\tol}} \right\} \cap A_m \mid N \bigg) \\
\label{piece2}
&+  \P\bigg( \left\{ \sup_{u \in (0,b')}b' \left|\frac{m}{N}\cdot G_m(u) -\frac{u}{b'} \right| > \frac{|f'(t_{\tol})| h^2k^2m^{-2/3}}{32\hat{\tol}} \right\} \cap A_m \mid N \bigg) \bigg].
\end{align}
Since $\frac{m}{N} G_m(u) \leq 1$ for any $u \in (0,b')$, the first term (\ref{piece1}) is bounded
\begin{align*}
(\ref{piece1}) &\leq \P\bigg( \left| \frac{N}{m b'} - 1\right|  > \frac{|f'(t_{\tol})| h^2k^2m^{-2/3}}{32(5\tol/4)b'} \bigg) \tag{$\hat{\tol} \leq 5\tol/4$ on $A_m$} \\
&\leq 2 \exp\left(-\frac{1}{3} \cdot mb' \cdot \left(\frac{|f'(t_{\tol})| h^2k^2m^{-2/3}}{40 \tol b'} \right)^2 \right) \tag{Binomial tail bound}\\
&= \exp\left(-\frac{1}{3} \cdot \frac{f'(t_{\tol})^2 h^3 k^4}{40^2 \tol(k+1)}  \right) \tag{$b'= \tol^{-1}h(k+1)m^{-1/3}$}.
\end{align*}
For (\ref{piece2}), note that conditional on $N$, the $U_{(1)}\leq \dots \leq U_{(N)}$ are equal in distribution to the order statistics of a size $N$ sample from the Uniform$(0,b')$ distribution, and apply the DKW inequality to obtain
\begin{align*}
(\ref{piece2}) &\leq 2\E \exp\left(-2 N \cdot \frac{f'(t_{\tol})^2 h^4 k^4 m^{-4/3}}{40^2 \tol^2(b')^2} \right), \tag{DKW} \\
&= 2\E \exp\left( -\frac{2f'(t_{\tol})^2 h^2 k^4 m^{-2/3}}{ 40^2(k+1)^2 } \cdot N\right)
\end{align*}
Since $N \sim \text{Binomial}(m,b')$, the above is
\begin{align*}
2 \left(1-b'+b'\cdot e^{ -\frac{2f'(t_{\tol})^2 h^2 k^4 m^{-2/3}}{ 40^2 (k+1)^2 }} \right)^m \leq 2\exp\left(mb' \left(e^{ -\frac{2f'(t_{\tol})^2 h^2 k^4 m^{-2/3}}{ 40^2 (k+1)^2 }}-1\right) \right).
\end{align*}
Since $k \leq h^{-1}m^{1/9}$, the exponent is eventually greater than $-1/2$ for any $h$ greater than a constant, so by the inequality $e^{x}-1 \leq x/2$ for $x>-1/2$, the above is further bounded by
\begin{align*}
 \leq 2\exp\left(mb' \left( -\frac{f'(t_{\tol})^2 h^2 k^4 m^{-2/3}}{ 40^2 (k+1)^2 } \right) \right) = 2 \exp\left( - \frac{f'(t_{\tol})^2 h^3 k^4 }{40^2 \tol(k+1)} \right).
\end{align*}
Combining the bounds on (\ref{piece1}) and (\ref{piece2}), the sum over $k$ in (\ref{union-bd}) becomes
\begin{align*}
(\ref{union-bd}) &\leq \sum_{k=1}^{\lceil h^{-1}m^{1/9}\rceil} \left[ \exp\left( -\frac{1}{3} \cdot \frac{f'(t_{\tol})^2 h^3 k^4 }{40^2 \tol(k+1)} \right) + 2 \exp\left( -\frac{f'(t_{\tol})^2 h^3 k^4 }{40^2 \tol(k+1)} \right)\right] \\
&\leq 3\sum_{k=1}^\infty \exp\left( -\frac{1}{3} \cdot \frac{f'(t_{\tol})^2 h^3 }{40^2 \tol} \cdot \frac{1}{2}\cdot k^3 \right) \\
&\leq 3\cdot \frac{\exp\left( - \frac{f'(t_{\tol})^2 h^3 }{6\cdot 40^2 \tol} \right)}{1-\exp\left( - \frac{f'(t_{\tol})^2 h^3 }{6\cdot 40^2 \tol} \right)} ,
\end{align*}
by the formula for a geometric series. Integrating this against $h$ gives a finite quantity. For large enough $M$, the integral of this bound (against $h$) from $M$ to $\infty$ is small enough that the second term in (\ref{tail-probs}) is less than $\delta/2$. Similarly, $M$ can be taken so large that this bound implies the first term in (\ref{tail-probs}) is less than $\delta/2$.
\end{proof}

We now turn to the asymptotics for the threshold~$\tau_{\hat\tol}$ and the sample maximum-$\lfdr$. Our proof will use the {\em switching relation} that states, for any $t\in (0,1)$, we have almost surely
\[
\tau_{\hat\tol} \leq t \;\iff\; \gren(t) \leq \hat{\tol}^{-1}.
\]

\begin{proof}[Proof of Theorem~\ref{thm-asymptotics}] 
We will work with a local expansion of $\gren(t)$ around $t_{\tol}$ using the local parameterization $t = t_{\tol} + m^{-1/3} h$. Using 
$f(t_{\tol}) = \tol^{-1}$, the switching relation becomes
\[
m^{-1/3}(\tau_{\hat\tol} - t_{\tol}) \leq h 
\;\iff\; 
\gren\left(t_{\tol} + m^{-1/3}h\right) - f(t_{\tol}) \leq
\hat{\tol}^{-1}-\tol^{-1}.
\]
Now let $W$ denote a standard two-sided Brownian motion, and let $\mathbb{S}_{a,b}$ denote the process of left derivatives of the least concave majorant of $X_{a,b}(t) = aW(t) - bt^2$, where $a = \sqrt{f(t_{\tol})}$ and $b = |f'(t_{\tol})|/2$. Under our regularity assumptions, the introduction of \citet{dumbgen2016law} provides
\[
m^{1/3}\left(\gren\left(t_{\tol} + m^{-1/3}h\right) - f(t_{\tol})\right)\Rightarrow \mathbb{S}_{a,b}(h)
\]
in the Skorokhod topology on $D[-K, K]$ for every finite $K > 0$. Since~$m^{1/3}(\hat\tol^{-1} -\tol^{-1})\stackrel{p}{\to} 0$ by assumption, we have
\[
\P\left\{m^{1/3}\left(\tau_{\hat\tol} - t_{\tol}\right)\le h\right\}
\to  \P\left\{\mathbb{S}_{a,b}(h) \le 0\right\}.
\]
Observe that $\mathbb{S}_{a,b}(h) \le 0$ iff $t_{a,b}^*\le h$, where $t_{a,b}^*$ is the (a.s. unique) maximizer of $X_{a,b}$ (note the maximizer $t_{a,b}^*$ is always a knot in the concave majorant since the horizontal line with intercept $X_{a,b}(t_{a,b}^*)$ is a supporting line intersecting $(t_{a,b}^*,X_{a,b}(t_{a,b}^*))$). Combining this observation with the previous display, we have
\[
m^{1/3}\left(\tau_{\hat{\tol}} - t_{\tol}\right) \stackrel{d}{\to} t_{a,b}^* 
\;\stackrel{d}{=}\; \left(b/a\right)^{-2/3} Z
\;=\; \left(\frac{\tol}{4}\cdot f'(t_{\tol})^2\right)^{-1/3} Z,
\]
proving~\eqref{eq-threshold-asymptotics}. Next we turn to the $\lfdr$ asymptotics. By Taylor's theorem,
\begin{align*}
m^{1/3}\left(\lfdr(\tau_{\hat\tol}) - \pi_0\tol\right)
&= \lfdr'(s) \cdot m^{1/3}\left(\tau_{\hat\tol} - t_{\tol}\right) 
\end{align*}
for some~$s$ between $\tau_{\hat{\tol}}$ and $t_{\tol}$. Using
\[
 \lfdr'(t_{\tol}) \;=\; \frac{-\pi_0 f'(t_{\tol})}{f(t_{\tol})^2} \;=\;
 \pi_0\tol^2\cdot |f'(t_{\tol})|,
\]
and applying the continuous mapping theorem and Slutsky's theorem, we obtain
\[
    \lfdr'(s) \cdot m^{1/3}\left(\tau_{\hat\tol} - t_{\tol}\right) 
    \;\stackrel{d}{\to}\; \lfdr'(t_{\tol})\cdot  \left(\frac{\tol}{4}\cdot f'(t_{\tol})^2\right)^{-1/3} Z
    \;=\; \pi_0\tol \cdot \left(4\tol^2\cdot |f'(t_{\tol})|\right)^{1/3} Z,
\]
proving~\eqref{eq-lfdr-asymptotics}. Next, under the strengthened assumption~\eqref{eq-stronger}, fix $\varepsilon > 0$ and define the event
\begin{equation}\label{eq-truncation}
A_{\varepsilon} = \left\{ |\hat\tol -\tol| \leq m^{-1/3}\varepsilon, \;|\tau_{\hat{\tol}} - t_{\tol}| \leq m^{-2/9}\right\},
\end{equation}
and the truncated random variable
\[
Y_m = m^{1/3}(\tau_{\hat\tol}-t_{\tol}) \cdot 1_{A_\varepsilon},
\]
We will show that $\P(A_{\varepsilon}^c) = o\left(m^{-2/3}\right)$. As a result, $Y_m$ has the same limit in distribution as $m^{1/3}(\tau_{\hat\tol}-t_{\tol})$. By Lemma \ref{lemma-uniform-integrability}, the sequence $Y_m^2$ is uniformly integrable, yielding convergence of the mean and variance of $Y_m$ to the mean and variance of its limiting distribution. Then, because
\[
\E\left[\left( m^{1/3}(\tau_{\hat\tol}-t_{\tol}) - Y_m\right)^2\right] \;\leq\; m^{2/3} \P(A_{\varepsilon}^c) \;\to\; 0,
\]
we will have the same limiting mean and variance for $m^{1/3}(\tau_{\hat\tol} - t_{\tol})$.

To show that $\P(A_{\varepsilon}^c) = o\left(m^{-2/3}\right)$, let $\tol_1 =\tol - m^{-1/3}\varepsilon$ and $\tol_2 =\tol + m^{-1/3}\varepsilon$ and assume that $m$ is sufficiently large that $m^{-1/3}\varepsilon \leq m^{-2/9}/2$, and
\[
f'(t) \leq f'(t_{\tol})/2, \quad \text{ for all } \; t \in [t_{\tol} - m^{-2/9}, \;t_{\tol} + m^{-2/9}].
\]
As a result, for all $t \geq t_{\tol_2} + m^{-2/9}/2$, we have
\begin{align*}
F(t) - F(t_{\tol_2}) - \frac{t-t_{\tol_2}}{\tol_2} 
&\;\leq\;  F(t_{\tol_2} + m^{-2/9}/2) - F(t_{\tol_2}) - \frac{m^{-2/9}}{2\tol_2} \\
&\;\leq\; \frac{f'(t_{\tol})}{16}\cdot m^{-4/9}
\end{align*}
Then, since $\tau_{\hat{\tol}} \leq \tau_{\tol_2}$ a.s. on $A_{\varepsilon}$, we have
\begin{align*}
\P\left\{\tau_{\hat{\tol}} > t_{\tol} + m^{-2/9}, \;A_{\varepsilon}\right\}
&\;\leq\; \P\left\{\tau_{\tol_2} > t_{\tol_2} + m^{-2/9}/2\right\}\\[7pt]
&\;\leq\; \P\left\{\sup_{t \geq t_{\tol_2} + m^{-2/9}/2} F_m(t) - F_m(t_{\tol_2}) - \frac{t-t_{\tol_2}}{\tol_2} \geq 0\right\}\\[7pt]
&\;\leq\; \P\bigg\{\sup_{t \geq t_{\tol_2} + m^{-2/9}/2} 
\Big(F_m(t) - F(t) \\
&~~~~~~~~~~~~~~~~~~~~~~~~~~~~~~~~~- \left(F_m(t_{\tol_2}) - F(t_{\tol_2})\right)\Big) \geq \frac{|f'(t_{\tol})|}{16}\cdot m^{-4/9}\bigg\}\\[7pt]
&\;\leq\; \P\left\{\sup_{t \in [0,1]} |F_m(t) - F(t)| \geq \frac{|f'(t_{\tol})|}{32}\cdot m^{-4/9}\right\}\\[7pt]
&\;\leq\; 2 \exp\left\{-\frac{f'(t_{\tol})^2}{512}\cdot m^{1/9}\right\},
\end{align*}
by the Dvoretzky--Kiefer--Wolfowitz inequality. An analogous argument yields the same bound for $\P\{\tau_{\hat{\tol}} \leq t_{\tol} - m^{-2/9}\}$.
\end{proof}

Next, we derive rates for the regret under weighted classification loss.

\begin{proof}[Proof of Theorem~\ref{thm-regret}] Define $\tol=\alpha/\pi_0$ and $\hat{\tol}=\alpha/\hat{\pi}_0$, and let $\Delta \subseteq \{1,\ldots,m\}$ denote the symmetric difference between the two rejection sets:
\[
\Delta = \begin{cases} 
\{R_{\hat{\tol}} + 1, \ldots, R^*\} & \text{if } R_{\hat{\tol}} < R^*\\
\{R^* + 1, \ldots, R_{\hat{\tol}}\} & \text{if } R_{\hat{\tol}} > R^*\\
\emptyset & \text{if } R_{\hat{\tol}} = R^*
\end{cases}.
\]
Then we have
\[
L_\omega(H, \cR_{\hat{\tol}}) - L_\omega(H, \cR^*) \;=\; \frac{1}{m} \left( R^* - R_{\hat{\tol}} + \frac{\text{sgn}(R_{\hat{\tol}} - R^*)}{\alpha}\sum_{i\in \Delta} (1-H_i)\right).
\]
Conditional on $F_m$, we have $H_i \simind \text{Bern}(1-\lfdr(p_{(i)}))$, giving conditional expectation
\begin{align*}
\Gamma_m 
&\;\coloneqq \; \E\bigg[L_\omega(H, \cR_{\hat{\tol}}) - L_\omega(H, \cR^*)\mid F_m\bigg]\\[7pt]
&\;=\; \frac{1}{m} \left( R^* - R_{\hat{\tol}} + \frac{\text{sgn}(R_{\hat{\tol}} - R^*)}{\alpha}\sum_{i\in \Delta} \lfdr(p_{(i)})\right)\\[7pt]
&\;=\; \int_{\tau_{\hat{\tol}}}^{\tau^*} \left(1-\alpha^{-1}\lfdr(t)\right)\textnormal{d}F_m(t)\\[7pt]
&\;=\; \rho(\tau_{\hat{\tol}}) + \alpha^{-1}\int_{\tau_{\hat{\tol}}}^{\tau^*} \left(\alpha - \lfdr(u)\right)(\textnormal{d}F_m(u) - \textnormal{d}F(u))
\end{align*}
Define the same truncation event $A_\varepsilon$ as in \eqref{eq-truncation}:
\[
A_{\varepsilon} = \left\{ |\hat\tol -\tol| \leq m^{-1/3}\varepsilon, \;|\tau_{\hat{\tol}} - \tau^*| \leq m^{-2/9}\right\}.
\]
Then, because $|\Gamma_m| \leq \alpha^{-1}$ we have
\begin{equation}\label{eq-bigbound}
\begin{aligned}
\bigg|\Regret_m(\cR_{\hat{\tol}}) &- \E\left[\rho(\tau_{\hat{\tol}}) 1_{A_{\varepsilon}}\right]\bigg| \\
&\;\leq\; \alpha^{-1}\E\left[\left|\int_{\tau_{\hat{\tol}}}^{\tau^*} \left(\alpha - \lfdr(u)\right)(\textnormal{d}F_m(u) - \textnormal{d}F(u))\right| 1_{A_{\varepsilon}} \right] + \alpha^{-1}\P\left(A_{\varepsilon}^c\right).
\end{aligned}
\end{equation}
We showed in the proof of Theorem~\ref{thm-asymptotics} that $\P\left(A_{\varepsilon}^c\right) = o\left(m^{-2/3}\right)$. Furthermore,
\begin{align*}
    m^{2/3} \E\left[\rho(\tau_{\hat{\tol}}) 1_{A_{\varepsilon}}\right] 
    &\;=\; \E\left[ \frac{f'(\xi_{\tau_{\hat{\tol}}})}{2} \cdot m^{2/3} (\tau_{\hat{\tol}} - \tau^*)^2 \cdot 1_{A_{\varepsilon}} \right]\\
    &\;\to\; \frac{f'(\tau^*)}{2} \left(\frac{\alpha}{4\pi_0}\cdot f'(\tau^*)^2\right)^{-2/3} \text{Var}(Z)\\[7pt]
    &\;=\; \left(\frac{\alpha^2}{2\pi_0^2}\cdot |f'(\tau^*)|\right)^{-1/3} \text{Var}(Z),
\end{align*}
where we have used the fact that $f'(\xi_{\tau_{\hat{\tol}}})$ is uniformly close to $f'(\tau^*)$ on $A_{\varepsilon}$. 

It remains only to show that the first term on the right-hand side of \eqref{eq-bigbound} is $o\left(m^{-2/3}\right)$. On~$A_\varepsilon$,~$|\tau^*-\tau_{\hat{\tol}}|\le m^{-2/9}$, so
\begin{align*}
m^{2/3}\E&\left[\left|\int_{\tau_{\hat{\tol}}}^{\tau^*} \left(1 -\alpha^{-1} \lfdr(u)\right)(\textnormal{d}F_m(u) - \textnormal{d}F(u))\right| 1_{A_{\varepsilon}} \right] \\
&\le \alpha^{-1}\E\left[m^{2/3}\sup_{t : |t-\tau^*|\le m^{-2/9}}\left|\int_{\tau^*}^{t} \left(\lfdr(u) - \alpha\right)(\textnormal{d}F_m(u) - \textnormal{d}F(u))\right| \right]
\end{align*}
The integrand~$g(u) = \lfdr(u)-\alpha$ is positive and increasing for~$u\ge \tau^*$. Furthermore, for~$m$ large we may bound~$g'(u)\le B$ uniformly on~$[\tau^*,\tau^*+m^{-2/9}]$, so that~$g(u) \le Bm^{-2/9}$. Discretize the upper range~$[\tau^*, \tau^*+m^{-2/9}]$ into bins of width~$w$ by~$t_0=\tau^*,\ldots,t_L=\tau^*+Lw$, where $L\coloneqq\lceil \frac{m^{-2/9}}{w}\rceil$. For~$t\in[t_{l-1}, t_l]$, 
\begin{align*}
\int_{\tau^*}^{t} g(u)(\textnormal{d}F_m(u) - \textnormal{d}F(u))
&= \int_{\tau^*}^{t_l} g(u)(\textnormal{d}F_m(u) - \textnormal{d}F(u))
- \int_{t}^{t_l} g(u)(\textnormal{d}F_m(u) - \textnormal{d}F(u)) \\
&\le \int_{\tau^*}^{t_l} g(u)(\textnormal{d}F_m(u) - \textnormal{d}F(u))
+ \int_{t}^{t_l} g(u) \textnormal{d}F(u) \\
&\le \int_{\tau^*}^{t_l} g(u)(\textnormal{d}F_m(u) - \textnormal{d}F(u)) + (F(t_l) - F(t_{l-1}))Bm^{-2/9} \\
&\le \int_{\tau^*}^{t_l} g(u)(\textnormal{d}F_m(u) - \textnormal{d}F(u)) + \pi_0\alpha^{-1}Bm^{-2/9}w.
\end{align*}
Hence
\begin{align*}
\sup_{t-\tau^*\le m^{-2/9}}&\int_{\tau^*}^{t} g(u)(\textnormal{d}F_m(u) - \textnormal{d}F(u))\\
&\le \pi_0\alpha^{-1}Bm^{-2/9}w +\max_{l=0,\ldots,L}\int_{\tau^*}^{t_l} g(u)(\textnormal{d}F_m(u) - \textnormal{d}F(u)). 
\end{align*}
We control the tail of the finite maximum with a union bound and Chebyshev's inequality
\begin{equation}\label{eq-chebyshev-union}
\begin{aligned}
    \P&\left(\max_{l=0,\ldots,L}\int_{\tau^*}^{t_l} g(u)(\textnormal{d}F_m(u) - \textnormal{d}F(u)) \ge c\right) \\
    &~~~~~~~\le \frac{L+1}{c^2}\max_{l=0,\ldots,L}\textnormal{Var}\left(\int_{\tau^*}^{t_l} g(u)(\textnormal{d}F_m(u) - \textnormal{d}F(u))\right).
\end{aligned}
\end{equation}
Let $S = m\int_{\tau^*}^{t_l} g(u)\textnormal{d}F_m(u) = \sum_{i : p_i\in [\tau^*, t_l]}g(p_i)$. Conditioned on~$N = m(F_m(t_l) - F_m(\tau^*))$, the sum~$S$ has the same distribution as~$\widetilde{S} = \sum_{i=1}^Ng(\widetilde{p}_i)$ where~$\widetilde{p_i}\simiid$ with cdf~$\frac{F(\cdot)-F(\tau^*)}{F(t_l) - F(\tau^*)}$. Thus
\begin{align*}
    \textnormal{Var}\left(S\right)
    &= \E\left[\textnormal{Var}\left(S\mid N\right)\right]
    + \textnormal{Var}\left(\E\left[S\mid N\right]\right) \\
    &\le \E\left[N\textnormal{Var}\left(g(\widetilde{p}_i)\right)\right]
    + \textnormal{Var}\left(N\E\left[g(\widetilde{p}_i)\right]\right) 
    \le B^2m^{-4/9}\E N + B^2m^{-4/9}\textnormal{Var}(N) \\
    &\le 2B^2m^{5/9}(F(t_l) - F(\tau^*)) \le 2\pi_0\alpha^{-1}B^2m^{1/3}.
\end{align*}
From this bound on the variance,~\eqref{eq-chebyshev-union} becomes
\begin{align*}
\P\left(\max_{l=0,\ldots,L}\int_{\tau^*}^{t_l} g(u)(\textnormal{d}F_m(u) - \textnormal{d}F(u)) \ge c\right) 
&\le \frac{L+1}{c^2}2\pi_0\alpha^{-1}B^2m^{-5/3} \\
&\le 4\pi_0\alpha^{-1}B^2\frac{m^{-17/9}}{wc^2}.
\end{align*}
We bound the expectation by integrating the tail:
\begin{align*}
    \E&\left[\sup_{t-\tau^*\le m^{-2/9}}\int_{\tau^*}^{t} g(u)(\textnormal{d}F_m(u) - \textnormal{d}F(u))\right]\\
    &~~~~~~~\le \pi_0\alpha^{-1}Bm^{-2/9}w +\int_0^\infty \P\left\{\max_{l=0,\ldots,L}\int_{\tau^*}^{t_l} g(u)(\textnormal{d}F_m(u) - \textnormal{d}F(u))\ge c\right\}\textnormal{d}c \\
    &~~~~~~~\le \pi_0\alpha^{-1}Bm^{-2/9}w +\int_0^\infty\min\left\{1, 4\pi_0\alpha^{-1}B^2\frac{m^{-17/9}}{wc^2}\right\}\textnormal{d}c \\
    &~~~~~~~\le \pi_0\alpha^{-1}Bm^{-2/9}w + \sqrt{4\pi_0\alpha^{-1}B^2\frac{m^{-17/9}}{w}} \\
    &~~~~~~~~~~~~~~~~~~~~~~+\int_{\sqrt{4\pi_0\alpha^{-1}B^2\frac{m^{-17/9}}{w}}}^\infty 4\pi_0\alpha^{-1}B^2\frac{m^{-17/9}}{wc^2}\textnormal{d}c \\
    &~~~~~~~\le \pi_0\alpha^{-1}Bm^{-2/9}w + 4\sqrt{\pi_0\alpha^{-1}B^2\frac{m^{-17/9}}{w}}
\end{align*}
Setting~$w = \left(16\alpha\frac{m^{-13/9}}{\pi_0}\right)^{1/3}$,
\begin{align*}
m^{2/3}\E\left[\sup_{t-\tau^*\le m^{-2/9}}\int_{\tau^*}^{t} g(u)(\textnormal{d}F_m(u) - \textnormal{d}F(u))\right]
&\le 2B\left(4\pi_0\alpha^{-1}\right)^{2/3} m^{2/3}m^{-13/27}m^{-2/9} \\
&= O(m^{-1/27}).
\end{align*}
The supremum over the range~$[\tau^*-m^{-2/9},\tau^*]$ is handled similarly.
\end{proof}

Next, we derive an exact, finite-sample expression for the regret of the SL procedure under the global null, and we use this result to show that the regret is~$O(m^{-1})$ in this case.

\begin{proof}[Proof of Proposition~\ref{prop-global-null}] Since~$H_i = 0$ for all~$i$
\[
L_\omega(H, \cR_\alpha)
- L_\omega(H, \cR^{\textnormal{OPT}}_\alpha)
= \frac{\omega \hat{R}_\alpha}{m}.
\]
Recall~$\hat{R}_\alpha$ is the argmax of the random walk~$k\mapsto \alpha\frac{k}{m} - p_{(k)}$, which has exchangeable increments. We will use Corollary 11.14 of \citet{kallenberg2002foundations}, due to Sparre-Andersen, that, by exchangeability, the number of rejections~$\hat{R}_\alpha$ is equal in distribution to the time the walk stays positive:
\[
\hat{R}_\alpha \stackrel{d}{=} P_\alpha \coloneqq \sum_{k=1}^m 1\left\{p_{(k)}\le \alpha\frac{k}{m}\right\}.
\]
Under the global null, the regret thus has mean
\begin{align*}
m\E\left[L_\omega(H, \cR_\alpha)
- L_\omega(H, \cR^{\textnormal{OPT}}_\alpha)\right]
&= \omega\E\hat{R}_\alpha 
= \omega \sum_{k=1}^m\P\left\{p_{(k)} \le \alpha\frac{k}{m}\right\} \\
&\to 
\omega \sum_{k=1}^\infty\P_{U_k\sim \textnormal{Gamma}(k, k)}\left\{U_k \le \alpha\right\}, 
\end{align*}
where the last step follows from the law of rare events.
\end{proof}

\edit{
Finally, we show consistency of the SL threshold, relaxing the independence assumption.

\begin{proof}[Proof of Proposition~\ref{prop-dependence}] The threshold $t_{\tol}$ is the unique minimizer of the convex function $H(t) = t-\tol F(t)$; similarly $\tau_{\tol}$ is the largest argmin of the random convex function $H_m(t) = t-\tol\hat{F}_m(t)$, where $\hat{F}_m$ denotes the LCM of $F_m$. Because $F$ is concave, we know \[|H_m-H|_\infty =\tol|\hat{F}_m - F|_\infty\le\tol|F_m - F|_\infty\] by Marshall's inequality \citep{marshall1970discussion}. Hence $|H_m-H|_\infty\stackrel{p}{\to}0$ follows from our assumption that $|F_m - F|_\infty\stackrel{p}{\to}0$.

Define $\delta(\eps) = \sup\{|t-t_\tol|:\; H(t) \leq H(t_\tol) + 2\eps\}$. If $|t_\tol - \tau_\tol| > \delta(\eps)$, we must have $H(\tau_\tol) > H(t_\tol) + 2\eps$, which further implies
\[
H_m(\tau_\tol) + 2|H_m-H|_\infty > H_m(t_\tol)  + 2\eps.
\]
Since $\tau_\tol$ is a minimizer of $H_m$, this implies $|H_m-H|_\infty>\eps$. We conclude that, for any $\eps >0$,
\[
\P\left\{|\tau_\tol - t_\tol| > \delta(\eps)\right\}
\le \P\left\{|H_m-H|_\infty > \eps\right\}
\to 0 \qquad \text{as } m\to\infty.
\]
It remains to be shown that $\delta(\eps)\downarrow 0$ as $\eps\downarrow0$. We will show this in the case that~$t_\tol \neq 0$ or $1$. Fix any $\delta_0 \in (0, t_\tol\land 1-t_\tol)$. Let $t_- = t_\tol - \delta_0$, $t_+ = t_\tol + \delta_0$, and 
\[
\eps_0 = \frac{1}{2}\left(H(t_-)\land H(t_+) - H(t_\tol)\right).
\]
Note that $\eps_0 > 0$ since~$t_\tol$ is the unique minimizer of~$H$. By construction, $H(t) \ge H(t_\tol) + 2\eps_0$ for both~$t\in \{t_-, t_+\}$, with equality for at least one of~$t_-, t_+$. By convexity, we must have $H(t) > H(t_\tol) + 2\eps_0$ for all $t \not\in [t_-, t_+]$, so $\delta(\eps_0) = \delta_0$. Since~$\delta(\eps)$ is non-decreasing in~$\eps$, this completes the proof that $\lim_{\eps\downarrow 0}\delta(\eps) = 0$. The case where $t_\tol = 0$ or $1$ is proved similarly.
\end{proof}
}

\edit{\section{Supplementary numerical results}

In this section, we assess how well the maximum $\max_{i\in \cR}\lfdr(p_i)$ concentrates around its expectation $\maxlfdr(\cR)$ (analogous to Figure~\ref{fig-lfdr-asymptotics}) under the various violations of assumptions considered in Section~\ref{sec-robustness}. Figure~\ref{fig-lfdr-asymptotics-robustness} shows the interquartile range of~$\max_{i\in \cR}\lfdr(p_i)$ across $10^5$ simulation runs with~$m=64$ or~$m=1024$ hypotheses. The top row shows results for the equicorrelated model; the middle row, the autoregressive model; and the bottom row, the misspecified model. For $m=1024$, the maximum concentrates well in most cases.
}

\begin{figure}[p!]
	\centering
    \centerline{
	\includegraphics[height=.29\textheight]{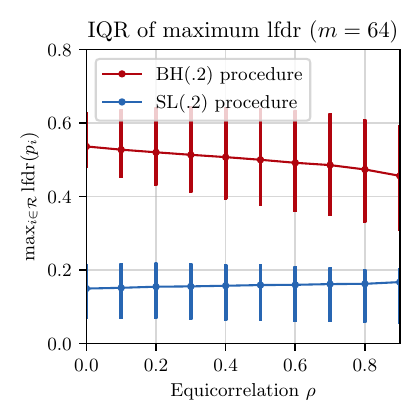}
	\includegraphics[height=.29\textheight]{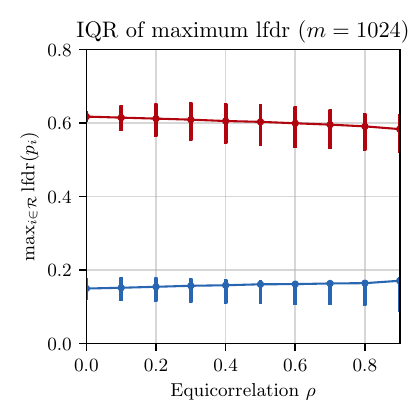}
	}
	\centerline{
	\includegraphics[height=.29\textheight]{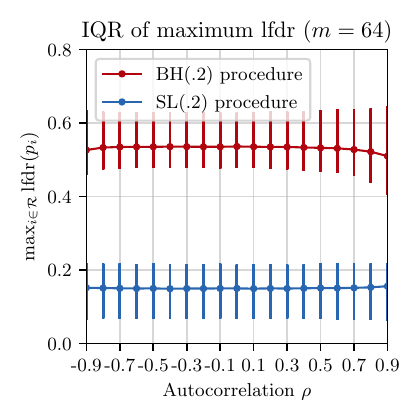}
	\includegraphics[height=.29\textheight]{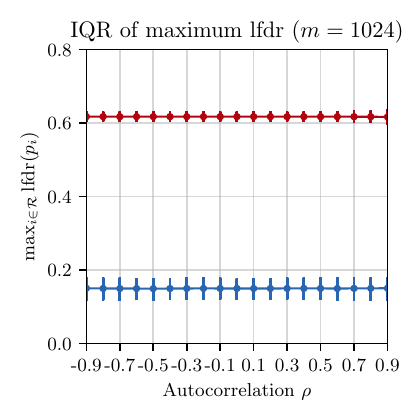}
	}
    \centerline{
	\includegraphics[height=.29\textheight]{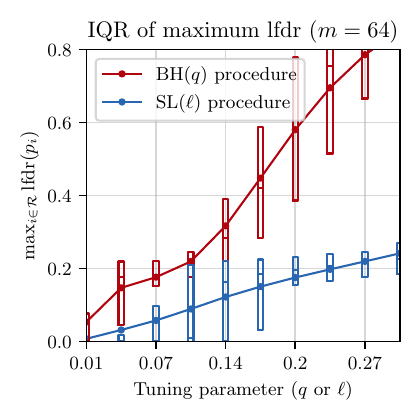}
	\includegraphics[height=.29\textheight]{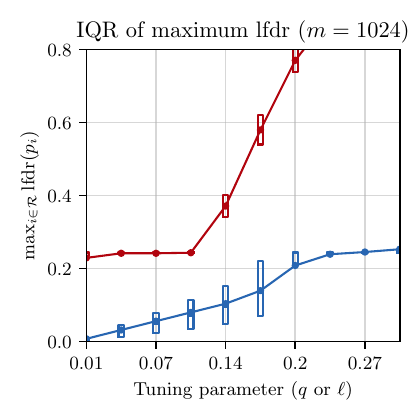}
	}
    \caption{\edit{Interquartile range of~$\max_{i\in \mathcal{R}}\lfdr(p_i)$ for the SL and BH procedures. Left column: $m=64$ hypotheses. Right column: $m=1,024$ hypotheses. Top row: results for equicorrelated model~\eqref{eq-equicorrelated-covariance}. Middle row: results for autoregressive model~\eqref{eq-autoregressive-covariance}. Bottom row: results for misspecified monotonicity constraint on the alternative density~\eqref{eq-misspecified-alternative}.}}\label{fig-lfdr-asymptotics-robustness}
\end{figure}

\end{appendix}

\FloatBarrier

\section*{Acknowledgements}

We are indebted to Lihua Lei for simplifying the proof of Lemma~\ref{lem-last-rejection}. We also thank Rina Foygel Barber, Stephen Bates, Aditya Guntuboyina, Michael I. Jordan, Peter McCullagh and Jim Pitman for insightful discussions. 

\section*{Funding}

J. A. S. was supported by the NSF [Grant DMS-2023505] and by a Vannevar Bush Faculty Fellowship [Grant N00014-21-1-2941]. W. F. was supported by the NSF [Grant DMS-1916220] and a Hellman Fellowship from Berkeley.

\nocite{shorack2009empirical, robbins1951asymptotically, takacs1967combinatorial, finner2001false}
\bibliographystyle{dcu}
\bibliography{reference.bib} 

\end{document}